\documentclass[11pt]{article}
\usepackage[margin=0.75in]{geometry}
\pdfoutput=1
\usepackage[utf8]{inputenc}
\usepackage[english]{babel}
\usepackage[T1]{fontenc}
\usepackage{amsmath,amssymb, amsthm}
\usepackage{graphicx}           
\usepackage{relsize}
\usepackage{tcolorbox}
\usepackage{amsthm}
\usepackage{thm-restate}
\usepackage{braket}
\usepackage{bbm}
\usepackage{cancel}
\usepackage{csquotes}
\usepackage{hyperref}
\usepackage{bm}
\usepackage{blkarray}
\usepackage{pdfpages}
\usepackage{array}
\usepackage{arydshln}
\usepackage{xcolor}
\usepackage{hhline}
\usepackage{multirow}

\usepackage{algorithm}
\usepackage{algpseudocode}
\algrenewcommand\algorithmicrequire{\textbf{Input:}}
\algrenewcommand\algorithmicensure{\textbf{Output:}}
\algnewcommand{\Initialize}[1]{
  \State\hspace*{-1em}\textbf{Initialize:}
  \Statex \hspace*{\algorithmicindent}\parbox[t]{.8\linewidth}{\raggedright #1}
}
\algnewcommand{\AlgoProcedure}[1]{\\
  \State\hspace*{-1em}\textbf{Procedure:}
  \Statex \hspace*{\algorithmicindent}\parbox[t]{.8\linewidth}{\raggedright #1}
}
\algnewcommand{\AlgoProcedureNoTitle}[1]{\\
  \Statex \hspace*{\algorithmicindent}\parbox[t]{.8\linewidth}{\raggedright #1}
}

\def\bstart#1\bstop{\boldsymbol{#1}}

\usepackage{soul}
\usepackage{mathtools}
\usepackage{cite}
\usepackage[figurename=Fig.,font=footnotesize]{caption}

\theoremstyle{definition}
\usepackage{float}
\makeatletter
\newlength\min@x
\makeatother

\usepackage{hyperref}
\hypersetup{
    colorlinks=true,
    linkcolor=blue,
    citecolor=magenta,
    filecolor=magenta,      
    urlcolor=blue,
    pdftitle={},
    pdfpagemode=FullScreen,
    }
\usepackage{cleveref}

\usepackage{lipsum}

\newtheorem{theorem}{Theorem}
\newtheorem{lemma}{Lemma}
\newtheorem{corollary}{Corollary}
\newtheorem{definition}{Definition}
\newtheorem*{remark}{Remark}
\newtheorem{proposition}{Proposition}
\newtheorem{example}{Example} 

\def\ket#1{| #1 \rangle}
\def\bra#1{\langle #1 |}

\title{Unified and Generalized Approach to Entanglement-Assisted Quantum Error Correction}
\author{Priya J. Nadkarni$^1$, Serge Adonsou$^{1,2}$, Guillaume Dauphinais$^1$,\\David W. Kribs$^{1,2}$\thanks{Corresponding author: \href{mailto:dkribs@uoguelph.ca}{dkribs@uoguelph.ca}}, and Michael Vasmer$^{1,3,4}$}
\date{
\small $^1$Xanadu, Toronto, ON M5G 2C8, Canada \\
$^2$Department of Mathematics \& Statistics, University of Guelph, Guelph, ON N1G 2W1, Canada \\
$^3$Perimeter Institute for Theoretical Physics, Waterloo, ON N2L 2Y5, Canada \\
$^4$Institute for Quantum Computing, University of Waterloo, Waterloo, ON N2L 3G1, Canada \\[1ex]
\normalsize \today
}

\begin{document}
\maketitle
\begin{abstract}
We introduce a framework for entanglement-assisted quantum error correcting codes that unifies the three original frameworks for such codes called EAQEC, EAOQEC, and EACQ under a single umbrella. The unification is arrived at by viewing entanglement-assisted codes from the operator algebra quantum error correction perspective, and it is built upon a recently established extension of the stabilizer formalism to that setting. We denote the framework by EAOAQEC, and we prove a general error correction theorem for such codes, derived from the algebraic perspective, that generalizes each of the earlier results. This leads us to a natural notion of distance for such codes, and we derive a number of distance results for subclasses of the codes. We show how EACQ codes form a proper subclass of the entanglement-assisted subspace codes defined by EAOAQEC. We identify and construct new classes of entanglement-assisted subsystem codes and entanglement-assisted hybrid classical-quantum codes that are found outside of the earlier approaches. 
\end{abstract}

\section{Introduction}
As quantum technologies continue to advance, the corresponding development of a variety of approaches for quantum error correction (QEC) will be necessary and surely play increasingly important roles as time goes on. The basic foundations for QEC were laid three decades ago \cite{PhysRevA.52.R2493,knill1997theory,PhysRevA.54.1098,ShorFaultTolerant}, and the subject now touches on all aspects of quantum information science. An important advance in the theory of QEC, which came roughly a decade later, was the introduction of quantum entanglement as a resource for boosting transmission rates when a sender and receiver share pre-existing entanglement \cite{brun2006correcting}.  

The resulting framework was built upon the stabilizer formalism for QEC \cite{gottesman1996class,gottesman1997stabilizer,PhysRevLett.78.405} and is appropriately called entanglement-assisted quantum error correction (EAQEC). Shortly thereafter it was generalized \cite{hsieh2007general} to the setting of operator quantum error correction (OQEC)  \cite{Kribs2005Unified,Kribs2006oqec} and subsystem codes \cite{Poulin2005Stabilizer,Bacon2006Operator,aly2008subsystem,Bombin2015Gauge,Bombin2015Single,PhysRevX.11.031039,kubica2022SingleshotQuantumError,Hastings2021dynamically,nemec2023quantum}, and denoted by the acronym EAOQEC. Subsequently, it was further generalized in a different direction to the setting of hybrid codes that simultaneously encode both classical and quantum information \cite{kuperberg2003capacity,devetak2005capacity,hsieh2010entanglement,hsieh2010trading,grassl2017codes,li2020error,cao2021higher,nemec2021infinite}, and that framework was denoted EACQ \cite{kremsky2008classical}. These entanglement-assisted approaches to error correction have now found applications in a wide variety of computational, experimental and theoretical settings, as a forward reference search on the papers \cite{brun2006correcting,hsieh2007general,kremsky2008classical} readily confirms. 

The QEC framework itself was generalized beyond the traditional and subsystem code cases, via a formulation that started with the Heisenberg picture for quantum dynamics instead of the Schr\"{o}dinger picture that is typically used. Coined operator algebra quantum error correction (OAQEC) \cite{Beny2007Generalization,Beny2007Quantum}, the approach allowed for the injection of operator algebra techniques more directly into the subject, leading to notions of von Neumann algebra codes and complementary private algebras, while also providing and building on techniques to correct hybrid classical-quantum and infinite-dimensional codes \cite{kuperberg2003capacity,devetak2005capacity,hsieh2010entanglement,hsieh2010trading,beny2009quantum,crann2016private}. More recently, interest in the OAQEC framework has been renewed through a mix of motivations, including hybrid coding theory   \cite{grassl2017codes,li2020error,cao2021higher,nemec2021infinite} and black hole theory  \cite{Verlinde2013,Almheiri2015,Harlow2017,hayden_penington2019,almheiri2018holographic,Penington2020,akers2019,kamal2019ryu,Akers2022,akers2022black}. These latest advances helped to provide the impetus to develop a stabilizer formalism for OAQEC, and very recently this was accomplished for finite-dimensional cases in \cite{kribs2023stabilizer}. A natural next step, which we address here, is to ask if there is an extension of entanglement-assisted quantum error correction to OAQEC? Such an extension would presumably generalize the previous approaches. Indeed, this was even left as an open line of investigation years ago in the wake of the original entanglement-assisted works (see the Conclusions of \cite{kremsky2008classical}). 

In this paper, we introduce a framework for entanglement-assisted quantum codes that simultaneously generalizes the three original frameworks given in \cite{brun2006correcting,hsieh2007general,kremsky2008classical}. Our approach unifies the frameworks under a single umbrella by viewing them through the lens of OAQEC, and it is built upon the corresponding extension of the stabilizer formalism to that setting from \cite{kribs2023stabilizer}. We denote the framework by EAOAQEC, and we prove a general error correction theorem for such codes, derived from the algebraic perspective, that generalizes each of the earlier results. This in turn leads to a notion of distance for such codes, and we use it to derive a number of distance results for subclasses of the codes. We show how EACQ codes form a proper subclass of the entanglement-assisted subspace codes defined by the framework. Further, we identify and construct new classes of entanglement-assisted subsystem codes and entanglement-assisted hybrid classical-quantum codes that are not captured by the earlier approaches but naturally find a place in EAOAQEC. 

This paper is organized as follows. In the next section we establish notation and briefly recall the elements of the stabilizer formalism we require. In \Cref{sec:EAOAQEC} we present the EAOAQEC framework and in \Cref{sec:ErrorCorrectionTheorem} we establish the general error correction theorem and code distance definition. \Cref{sec:EAOAQEC_special_cases} indicates how EAQEC and EAOQEC are captured as special cases (and we include a schematic diagram to keep track of the different acronyms). In \Cref{sec:eacq} we show how EACQ codes form a subclass of EAOAQEC subspace codes, and we identify algebraic conditions that specify the subclass. In \Cref{sec:eaoaqec_code_construction} we give a number of examples and constructions of EAOAQEC subsystem codes (both hybrid and non-hybrid). Some of the subsystem code constructions in this section are rather technical, and so we have placed parts of the arguments in appendices. We conclude in \Cref{sec:conc} with a summary and some forward looking remarks.  %

\section{Preliminaries}

We consider error correcting codes for sets of noise operators from the $n$-qubit Pauli group $\mathcal P_n$, which is the group of $n$-qubit unitary operators generated by $n$-tensors of the single qubit bit flip $X$ and phase flip $Z$ Pauli operators and $iI$ (we will write $I$ for the identity operator on any sized Hilbert space as the context will be clear). We use standard notation for $n$-qubit operators, such as $X_1 = X \otimes (I^{\otimes (n-1)})$, $X_2 = I\otimes  X \otimes (I^{\otimes (n-2)})$, etc. 

Each of the entanglement-assisted (EA) error correction approaches \cite{brun2006correcting,hsieh2007general,kremsky2008classical} is built upon a stabilizer formalism that constructs codes for Pauli noise models \cite{gottesman1996class,gottesman1997stabilizer,Poulin2005Stabilizer}. Here we make use of the codes constructed in the recently introduced stabilizer formalism \cite{kribs2023stabilizer} for `operator algebra quantum error correction' (OAQEC) \cite{Beny2007Generalization,Beny2007Quantum}. The starting point is the same as previous stabilizer formalism settings. 

Let $\mathcal S$ be an Abelian subgroup of $\mathcal P_n$ that does not contain $-I$, and suppose it has $s\geq 1$ independent generators. As an illustrative example for this discussion, take $\mathcal S = \langle Z_1, \ldots , Z_s \rangle$ to be the group generated by phase flip operators on the first $s$ qubits. The normalizer $\mathcal N(\mathcal S)$ and centralizer $\mathcal Z(\mathcal S)$ subgroups of $\mathcal S$ inside $\mathcal P_n$ coincide, as elements of the Pauli group either commute or anti-commute up to some power of $iI$ and $\mathcal S$ does not contain $\langle iI \rangle$. 
The stabilizer subspace for $\mathcal S$ is the joint eigenvalue-$1$ eigenspace for $\mathcal S$; that is, 
$
C =  C(\mathcal S) = \mathrm{span} \{ \ket{\psi} \, : \, g \ket{\psi} = \ket{\psi} \,\, \forall g\in \mathcal S \}. 
$
In the example,  $C(\mathcal S) = \mathrm{span} \big\{ \ket{ 0^{\otimes s} i_1 \cdots i_{n-s}  } \, : \, i_j = 0,1 \big\}$.  
We will let $P$ denote the codespace projector for $C$, which encodes $n-s$ qubits. 

Subsystem structure generated by a stabilizer subspace can be obtained as follows. 
Suppose we have subsets $\mathcal G_0$ and $\mathcal L_0$ of $\mathcal N(\mathcal S) = \mathcal Z(\mathcal S)$ with the following properties: (i) The `compression algebra' generated by $\mathcal G_0 P$, respectively $\mathcal L_0 P$,  is unitarily equivalent to a full matrix algebra $M_{2^r}$ for some $r \geq 1$, respectively to $M_{2^k}$ for some $k \geq 1$. 
(ii) The sets $\mathcal G_0$ and $\mathcal L_0$ are mutually commuting; $[g,L]=0$ for all $g\in \mathcal G_0$, $L\in \mathcal L_0$. 
(iii) $\mathcal N(\mathcal S) $ is generated by $\mathcal S$, $iI$, $\mathcal G_0$, and $\mathcal L_0$ (and so $n-s = r+k$). The trivial ancilla subsystem case (when $r=0$, and we will write $\mathcal G_0 = \emptyset$ the empty set in that case) can be viewed as a special case of this set up, and yields standard (subspace) stabilizer codes, but generally we view the subspace and bona fide subsystem cases (with $r\geq 1$ and $\mathcal G_0 \neq \emptyset$) as separate.  

The group $\mathcal G$ defined as
   $ \mathcal G = \langle \mathcal S, iI, \mathcal G_0 \rangle, $ 
is called the {\it gauge group} for the code, and the group
   $ \mathcal L = \langle \mathcal L_0, iI \rangle, $ 
is called the {\it logical group}. The third property above ensures that the normalizer satisfies the direct product group isomorphism $\mathcal N(\mathcal S) \times \langle iI \rangle \cong \mathcal G \times \mathcal L$. 
The codespace $C$ then decomposes as a tensor product of subsystems $C = A \otimes B$ with $A \cong (\mathbb C^2)^{\otimes r}$, $B \cong (\mathbb C^2)^{\otimes k}$, with $\mathcal G$ (respectively $\mathcal L$) restricted to $C$ generating $\mathcal L(A)\otimes I_B$ (respectively $I_A \otimes \mathcal L(B)$), where $\mathcal L(A)$ is the set of operators on $A$.  
For the example, we can take 
$\mathcal G_0 = \big\{  X_{s+1}, Z_{s+1}, \ldots , X_{s+r}, Z_{s+r} \big\}$ and 
$\mathcal L_0 = \big\{  X_{s + r +1}, Z_{s+ r + 1}, \ldots , X_{n}, Z_{n} \big\}$ to obtain these structures.    

Lastly, OAQEC stabilizer codes allow for hybrid classical-quantum encodings as follows. Let $\mathcal T \subseteq \mathcal P_n$  be a maximal set of coset representatives for $\mathcal N(\mathcal S)$, a so-called {\it coset transversal} for $\mathcal N(\mathcal S)$ as a subgroup of $\mathcal P_n$, and include $I\in \mathcal T$ as the representative for the normalizer coset itself. Then we have the (disjoint) union $\mathcal P_n = \cup_{g\in \mathcal T} \, g \mathcal N(\mathcal S) $, and the cardinality of $\mathcal T$ is equal to $|\mathcal T| = | \mathcal P_n | / |\mathcal N(\mathcal S) | = 2^s$. One can easily compute this directly for the example, for instance a choice of transversal in that case is given by the set  $\mathcal T = \big\{ X_1^{a_1} \cdots X_s^{a_s} \, : \, 0 \leq a_j \leq 1 \big\}$, as the set is multiplicatively closed and each operator $X_j$, $1\leq j \leq s$, does not belong to $\mathcal N(\mathcal S)$ nor do any of their products other than those that collapse to the identity. 

We use the terminology {\it code sector} to refer to the (quantum) code defined  by a given $T\in \mathcal T$ and the elements that define the base code: $\mathcal S$, $\mathcal L$, $\mathcal G$. So, the code sector for $T$ is defined by the collection of operators given by the sets $\{ T {\mathcal S} T^{-1}$, $T {\mathcal L} T^{-1}$, $T {\mathcal G} T^{-1} \}$, and the codespace $T C$. 
An important observation concerning normalizer cosets in this setting as noted in \cite{kribs2023stabilizer} is that the subgroup and coset  structure induces orthogonality at the Hilbert space level, in the sense that $\bra{\psi_1} g_1^{-1} g_2 \ket{\psi_2} = 0$ for all $g_1,g_2\in \mathcal T$ with $g_1\neq g_2$ and all $\ket{\psi_1}, \ket{\psi_2} \in C$; or equivalently, $P g_1^{-1} g_2 P =0$. 
Hence, any subset $\mathcal T_0 \subseteq \mathcal T$ defines an `OAQEC stabilizer code', which will be a hybrid classical-quantum code whenever $|\mathcal T_0| > 1$ and  $\dim C > 1$. 

To summarize, this formalism yields stabilizer codes that generalize both the original (subspace) setting of Gottesman \cite{gottesman1996class,gottesman1997stabilizer} (captured with $\mathcal T_0 = \{ I \}$ and $\mathcal G_0 = \emptyset$), and the OQEC (subsystem) setting of Poulin \cite{Poulin2005Stabilizer} (captured with $\mathcal T_0 = \{ I \}$ and $\mathcal G_0 \neq \emptyset$). As noted above, a code defined by $\mathcal T_0 \subseteq \mathcal T$ with $|\mathcal T_0|>1$ will be a hybrid code, with a subspace base code ($C = A \otimes B$ with $A = \mathbb C$) when the gauge group is Abelian and a subsystem base code ($C=A\otimes B$ with $\dim A > 1$) otherwise. The size of the subset $\mathcal T_0$ determines the number of classical bit strings that can be encoded in the code, each on top of a quantum code.
We point the reader to \cite{kribs2023stabilizer} for further details on the formalism. 

\section{Entanglement-Assisted Operator Algebra Quantum Error Correction}
\label{sec:EAOAQEC}

We will introduce key notions and notation while carrying an illustrative example through the discussion that includes all the key elements. We begin by recalling well-known (for instance see \cite{fattal2004entanglement,bravyi2006ghz}) basic structural features of Pauli subgroups used in entanglement-assisted code formulations.  

Suppose we have a (not necessarily Abelian) subgroup $\mathcal H$ of the $n$-qubit Pauli group $\mathcal P_n$. We can assume $\mathcal H$ has $2^m$ elements, up to overall phase, for some $1 \leq m \leq n$. Then there exists a set of $m$ independent (in a group-theoretic sense) generators for $\mathcal H$ of the form
\[
\big\{  \overline{Z}_1,  \ldots , \overline{Z}_{m-\ell}, \overline{Z}_{m-\ell + 1},  \ldots , \overline{Z}_\ell, \,\, \overline{X}_1,  \ldots , \overline{X}_{m-\ell} \big\}, 
\]
where $m/2 \leq \ell \leq m$, such that: 
\begin{center}
\begin{enumerate}
\item[(i)]  $[ \overline{Z}_i,  \overline{Z}_j ] = 0 =  [ \overline{X}_i,  \overline{X}_j ]$ for all $i,j$; 
\item[(ii)]  $[ \overline{Z}_i,  \overline{X}_j ]  = 0$ for all $ i \neq j$; 
\item[(iii)]  $\{ \overline{Z}_i,  \overline{X}_i \}  = 0$ for all $i$; 
\end{enumerate}
\end{center}
and where here $[A,B]$ and $\{A,B\}$ are the usual commutator and anti-commutator of $A$ and $B$. 

Given such a representation of generators, we can define two subgroups of $\mathcal H$: an {\it isotropic subgroup} generated by the commuting generators, 
\[
\mathcal H_I = \langle \overline{Z}_{m-\ell + 1},  \ldots , \overline{Z}_\ell\rangle , 
\]
and a {\it symplectic subgroup} generated by the set of anti-commuting generator pairs, 
\[
\mathcal H_S = \langle  \overline{Z}_1,  \ldots , \overline{Z}_{m-\ell},  \overline{X}_1,  \ldots , \overline{X}_{m-\ell}  \rangle .
\]
Then $\mathcal H = \langle \mathcal H_I , \mathcal H_S \rangle$ is generated by both subgroups. 
This isotropic-symplectic decomposition of $\mathcal H$ is not unique, but any two such decompositions are related by a unitary. More generally, we will make use of the fact that any group isomorphism between two Pauli subgroups $\mathcal H_1 \cong \mathcal H_2$ is unitarily implemented up to phases; that is, there is a unitary $U$ such that for all $h_1\in \mathcal H_1$, there is $h_2\in \mathcal H_2$ with $h_2 = U h_1 U^\dagger$ up to an overall phase. 

As an example, consider the following six generators for a 6-qubit Pauli subgroup $\mathcal H$ (this example is a variant of the original example from \cite{brun2006correcting}): 
\[
\begin{array}{c|cccccc}
\hhline{=======}
h_1 & Z & I & I & I & I & I \\ 
h_2 & X & I & I & I  & I & I \\ 
h_3 & I & Z & I & I & I & I \\ 
h_4 & I & X & I & I  & I & I \\ 
h_5 & I & I & Z & I  & I  & I \\ 
h_6 & I & I & I & Z & I & I \\
\hhline{=======}
\end{array}
\]
So here we have $n=6$, $m=6$, $\ell = 4$.  We can take the symplectic subgroup as generated by $\{ h_1, h_2, h_3, h_4\} = \{Z_1, X_1, Z_2, X_2\}$, and the isotropic subgroup as generated by $\{ h_5, h_6 \} = \{Z_3, Z_4\}$.  

As in previous EA settings, we can extend the generators by adding new qubits and operator actions on those qubits such that the new subgroup $\mathcal S$ on the extended space is Abelian. Typically this can be done in many ways, and in the traditional EAQEC settings, which we will follow here, it is done with a minimal number of qubits, which is equal to one additional {\it entangled bit} (ebit) for each non-commuting pair of generators in the original group $\mathcal H$. So the generators of the Abelian group $\mathcal S$ are the generators of $\mathcal H$ extended in this way. For the example above, this can be accomplished by adding $X$ and $Z$ on a pair of extra qubits to obtain the six generators of $\mathcal S$ as follows: 
\[
\begin{array}{c|cccccc|cc}
\hhline{=========}
S_1 & Z & I & I & I & I & I & Z & I \\ 
S_2 & X & I & I & I  & I & I & X & I \\ 
S_3 & I & Z & I & I & I & I & I & Z \\ 
S_4 & I & X & I & I  & I & I & I & X \\ 
S_5 & I & I & Z & I  & I  & I & I  & I \\ 
S_6 & I & I & I & Z & I & I & I  & I \\
\hhline{=========}
\end{array}
\]

We will let $e$ denote the number of ebits in an EA code; so $e=m-\ell$ in the general notation and $e=2$ in this example.  The number of what we will call {\it isotropic qubits} $s$ is equal to the number of independent generators of $\mathcal H_I$; so $s= 2 \ell - m$ in general and $s=2$ in the example (these are called `ancilla qubits' in early EAQEC formulations; here we use different terminology as we will need this term to refer to other qubits in our applications). Without any further structure, here we have an EAQEC (subspace) code that encodes $k = n - e -s = n- \ell$ qubits. In the example above, $k=2$ and the two-qubit code is given by the stabilizer subspace,  
\[
C : = C(\mathcal S) = \mathrm{span}\,\big\{ \ket{\Phi}_{1,7}\ket{\Phi}_{2,8}\ket{0}\ket{0}\ket{\Psi}   \},
\]
where $\ket{\Psi}$ is an arbitrary two-qubit state supported on the fifth and sixth qubits, $\ket{\Phi}_{1,7}$ is the standard maximally entangled Bell state between the first and seventh qubits, and similarly for $\ket{\Phi}_{2,8}$. 

If the $n-e-s = n - \ell$ qubits allow for a subsystem decomposition, say into a tensor decomposition of an $r$ qubit system with a $k$ qubit system, then we have a (OQEC) subsystem code with an entanglement-assisted feature, and this is how we can obtain EAOQEC codes. The overall $t = r+k$ qubit code subspace has $r$ {\it gauge qubits} along with $k$ encoded {\it logical qubits} (here we choose variables for the gauge and logical qubits in line with other settings). The {\it gauge group} $\mathcal G$ for the code is generated by the operators $\mathcal G_0$ on the gauge qubits, the stabilizer group $\mathcal S$ and the scalar phase group $\langle i I \rangle$. The {\it logical group} is generated by the logical operators $\mathcal L_0$ acting on the encoded qubits and  $\langle i I \rangle$. Note that by construction the logical and gauge operators are $(n+e)$-qubit operators that act as the identity operator on the extra $e$ ebits. We will use the notation $L^{(n)}$ to denote the operator acting on the first $n$-qubits for such operators; so for instance, $L= L^{(n)} \otimes I$ for $L\in \mathcal L_0$. In the example, we can choose $r=1=k$, and in the table below we identify the fifth and sixth qubits as the gauge and encoded qubits respectively, with gauge generators $\mathcal G_0 = \{\overline{G}_1, \overline{G}_2\}$ and $\mathcal L_0 = \{ \overline{L}_X, \overline{L}_Z \}$ as the logical operators acting on the encoded qubit: 
\[
\begin{array}{c|cccccc|cc}
\hhline{=========}
S_1 & Z & I & I & I & I & I & Z & I \\ 
S_2 & X & I & I & I  & I & I & X & I \\ 
S_3 & I & Z & I & I & I & I & I & Z \\ 
S_4 & I & X & I & I  & I & I & I & X \\ 
S_5 & I & I & Z & I  & I  & I & I  & I \\ 
S_6 & I & I & I & Z & I & I & I  & I \\
\hhline{---------}
\overline{G}_1 & I & I & I & I & X & I & I & I  \\  
\overline{G}_2 & I & I & I & I & Z & I  & I & I \\ 
\hhline{---------}
\overline{L}_X & I & I & I & I & I & X & I & I  \\  
\overline{L}_Z & I & I & I & I & I & Z  & I & I \\ 
\hhline{=========}
\end{array}
\]

Finally, we introduce hybrid classical-quantum structure by adopting the approach from \cite{kribs2023stabilizer}, through the coset structure of the subgroup $\mathcal N(\mathcal S)$ inside $\mathcal P_{n+e}$. Let $\mathcal T$ be a coset transversal for $\mathcal N(\mathcal S)$, which without loss of generality includes the identity operator, and so $\mathcal P_{n+e}$ is equal to the disjoint union of the sets $g\mathcal N(\mathcal S)$ with $g\in\mathcal T$. 
There are a total of $2^m$ cosets, where recall $m$ is the number of independent generators of $\mathcal H$, and hence of $\mathcal S$. The coset representatives and code space define a family of mutually orthogonal subspaces, $\{ g C \, : \, g\in\mathcal T \}$, each of which carries the structure of $C$ (which has subsystems if $r>1$) to $g C$ through the unitary action of $g$, along with the corresponding code sector operators $\{ g \mathcal S g^{-1}, g \mathcal G g^{-1}, g \mathcal L g^{-1} \}$. In particular, any subset of coset representatives $\mathcal T_0 \subseteq \mathcal T$ (which we will always assume includes $I$) with cardinality $|\mathcal T_0|  > 1$ defines a hybrid classical-quantum code.

Building on the example above, we can consider the hybrid code defined by the three coset representatives $\mathcal T_0 = \{ I, T_1, T_2\}$, where: 
\[
\begin{array}{c|cccccc|cc}
\hhline{=========}
S_1 & Z & I & I & I & I & I & Z & I \\ 
S_2 & X & I & I & I  & I & I & X & I \\ 
S_3 & I & Z & I & I & I & I & I & Z \\ 
S_4 & I & X & I & I  & I & I & I & X \\ 
S_5 & I & I & Z & I  & I  & I & I  & I \\ 
S_6 & I & I & I & Z & I & I & I  & I \\
\hhline{---------}
\overline{G}_1 & I & I & I & I & X & I & I & I  \\  
\overline{G}_2 & I & I & I & I & Z & I  & I & I \\ 
\hhline{---------}
\overline{L}_X & I & I & I & I & I & X & I & I  \\  
\overline{L}_Z & I & I & I & I & I & Z  & I & I \\ 
\hhline{---------}
T_0 & I & I & I & I & I & I  & I & I \\  
T_1 & I & I & X & I & I & I & I & I  \\  
T_2 & I & I & I & X & I & I  & I  &  I \\ 
\hhline{=========}
\end{array}
\]
It is easy to see that $\mathcal T_0$ defines three distinct cosets here. Indeed, $T_1$ (respectively $T_2$) does not commute with $S_5$ (respectively $S_6$) for instance, and so $T_i=T_i I \in T_i \mathcal N(\mathcal S)$ but $T_i \notin I \mathcal N(\mathcal S) = \mathcal Z(\mathcal S)$ for $i=1,2$, and so the cosets $T_i \mathcal N(\mathcal S)$, $i=1,2$, are distinct from $\mathcal N(\mathcal S)$. Similarly $T_1T_2$ does not commute with either of $S_5, S_6$, so the cosets $T_i \mathcal N(\mathcal S)$, $i=1,2$, are also distinct from each other. In fact, one can check that an example of a full transversal for this code is the $2^6$ element set: 
\[
\mathcal T = \big\{  X_1^{a_1} Z_1^{b_1} X_2^{a_2} Z_2^{b_2} X_3^{a_3} X_4^{a_4}\,\, | \,\, a_i, b_j \in \{0,1\}   \big\} . 
\]

Thus, we have constructed a class of entanglement-assisted codes, 
with structural features indicated by $C = C(\mathcal H, \mathcal S, \mathcal G_0, \mathcal L_0, \mathcal T_0)$. As a reminder, $\mathcal H$ is an $n$-qubit Pauli subgroup, $\mathcal S$ is an Abelian $(n+e)$-qubit Pauli subgroup with generators that extend the generators of $\mathcal H$ as in the discussion above, $\mathcal G_0$ and  $\mathcal L_0$ are the gauge and logical operators respectively for the code which act as the identity on the extended qubits, and $\mathcal T_0$ are the transversal operators (which we will see in the next section can also be chosen to act as the identity on the extended qubits). 

These codes are OAQEC codes on the extended space that arise from the stabilizer formalism for OAQEC, and they are EA codes when viewed on the original space. This justifies our terminological choice of entanglement-assisted operator algebra quantum error correcting (EAOAQEC) codes. We will use the following notation to denote the structural properties of these codes.  

\begin{definition}
    The parameters of an EAOAQEC code $C = C(\mathcal H, \mathcal S, \mathcal G_0, \mathcal L_0, \mathcal T_0)$ that encodes $k$ logical qubits and one of the $c_b$ classical bitstrings into $n$ physical qubits and $e$ ebits and contains $r$ gauge qubits is denoted by $[\![n,k; r, e, c_b]\!]$. At times we will suppress reference to features that are not present (e.g., $[\![n,k ; r, e]\!]$ for subsystem codes with no hybrid component and $[\![n,k ; r]\!]$ for subsystem codes with no hybrid component and no entanglement-assistance), and we will add $d$ for distance (defined below) when that notion is being considered for a code.
\end{definition}

\section{Error Correction Theorem} \label{sec:ErrorCorrectionTheorem}

In this section we find the general error correction conditions for EAOAQEC codes, which we will subsequently apply in multiple special cases, and we use them to define an appropriate notion of code distance. 

The traditional physical model for entanglement-assisted quantum error correction proceeds as follows. 
We assume that Alice and Bob have pre-shared $e$ ebits.
Alice then performs the unitary encoding operation on her $k$ qubits, her half of the ebits, and the $s$ isotropic qubits.
She then sends the $n$ qubits through a noisy channel to Bob.

A key assumption in the model is that the entangled qubits held by Bob are error-free, and so the errors to be considered are of the form $E \otimes I$ where $E$ is an $n$-qubit Pauli operator and $I$ is the identity operator on Bob's $e$ qubits. Bob can thus measure the generators of $\mathcal S$, which recall are extensions of the generators of the original group $\mathcal H$, on the full $n+e$ qubits, the outcome of which allows him to determine the error syndrome, correct the error and decode the transmitted information. 

This entire situation lifts to the full blown hybrid (subspace or subsystem) code setting, using the OAQEC notion of correctability as originally laid out in \cite{Beny2007Generalization,Beny2007Quantum}. In particular, we will make use of the main error correction result of \cite{kribs2023stabilizer} which characterizes what sets of Pauli errors are correctable for such a code. For background, that result generalized the stabilizer formalism theorems of Gottesman \cite{gottesman1996class,gottesman1997stabilizer} and Poulin \cite{Poulin2005Stabilizer}, and relied on the OAQEC testable conditions from \cite{Beny2007Generalization,Beny2007Quantum}, which in turn generalized the Knill-Laflamme \cite{knill1997theory} and OQEC conditions \cite{Kribs2005Unified,Kribs2006oqec}.  
Specifically, an $n$-qubit OAQEC stabilizer code $C = C(\mathcal S, \mathcal G_0, \mathcal L_0, \mathcal T_0)$, with $\mathcal T_0 = \{ T_i \}_i$,  is correctable for a set of error operators $\{F_a\} \subseteq \mathcal P_n$ if and only if for all $a,b$,  
\begin{equation}\label{stabformoaqeccond}
F_a^\dagger F_b \notin \Big( \mathcal N(\mathcal S)\setminus \mathcal G  \Big) \bigcup \Big(  \bigcup_{i\neq j} T_i T_j^{-1}  \mathcal N(\mathcal S)   \Big) .
\end{equation}

We shall apply this result to the present setting, and in particular $n$ will be replaced by $n+e$ and the $F_a$ will be of the form $E_a \otimes I$. Further, and as noted above, observe that a coset transversal $\mathcal T \subseteq \mathcal P_{n+e}$ for an EAOAQEC code can always be obtained with elements that are noiseless on Bob's qubits; that is, we can find $\mathcal T$ such that for all $T\in \mathcal T$, there is $T^{(n)} \in \mathcal P_n$ such that $T = T^{(n)} \otimes I$. Indeed, in the generic notation above, one can choose a set of $s$ elements $\{ \overline{X}_{m-\ell + 1}, \ldots , \overline{X}_\ell \}$ of $\mathcal P_n$ based on the elements of the initial group, that anti-commute with the corresponding isotropic generator exclusively and commute with every other generator of the original subgroup (these are also called \emph{destabilizers} when the isotropic generators are seen as generators for a stabilizer group \cite{PhysRevA.70.052328}.)  
Then the following $2^m$ elements will define a transversal for the hybrid code on the extended space: 
\[
\Big\{ (\overline{X}_{1}^{a_1}\otimes I) (\overline{Z}_{1}^{b_1}\otimes I) \cdots  (\overline{X}_{m - \ell}^{a_{m - \ell}} \otimes I) (\overline{Z}_{m - \ell}^{b_{m - \ell}}\otimes I) \,
(\overline{X}_{m-\ell + 1}^{a_{m - \ell + 1}}\otimes I) \cdots  (\overline{X}_\ell^{a_\ell} \otimes I) \, \, : \, \, a_j, b_j \in\{ 0 ,1\} \Big\}; 
\]
where this follows from the fact that any product of two of these operators does not belong to the normalizer of the extended stabilizer subgroup, and hence the operators define a set of distinct cosets of maximal size. (Recall that the extended Abelian subgroup has $m$ independent generators and hence its normalizer, inside the extended Pauli group, has $2^m$ cosets.)

As a consequence, this simplifies the encoding of these codes as it means that Alice and Bob do not need to use entangled operations to generate the hybrid features of the code. It also allows for an explicit characterization of correctable sets of errors as follows. 

\begin{theorem} \label{noiselesserrorcorrecthm}
Suppose we have an $n$-qubit entanglement-assisted code $C = C(\mathcal H, \mathcal S, \mathcal G_0, \mathcal L_0, \mathcal T_0)$, with $\mathcal T_0 = \{ T_i = T_i^{(n)} \otimes I \}_i$. Then a set of errors $\{ E_a \otimes I \}$ with $E_a \in \mathcal P_n$ is correctable if and only if for all $a,b$,   
\begin{equation}\label{noiselessBobeqn}
E_a^\dagger E_b \in \Big( \langle \mathcal H_I, \mathcal G_0^{(n)}, iI \rangle \bigcup \big( \mathcal P_n \setminus \mathcal Z(\mathcal H)\big) \Big) \bigcap \Big( \mathcal P_n \setminus \big( \bigcup_{i \neq  j}  T_i^{(n)} (T_j^{(n)})^{-1} \mathcal Z(\mathcal H)  \big)  \Big),
\end{equation}
where $\mathcal Z(\mathcal H)$ is the centralizer of $\mathcal{H}$ in $\mathcal{P}_n$.
\end{theorem}

\begin{proof}
This result comes as a direct application of the conditions given in Eq.~(\ref{stabformoaqeccond}). Note that we formulate the conditions of Eq.~(\ref{noiselessBobeqn}) as inclusions in the sets rather than set complements, in recognition of earlier entanglement-assisted error correction results as described in subsequent sections.  Let $E = E_a^\dagger E_b$, and note that $E \otimes I$ not belonging to the first set of Eq.~(\ref{stabformoaqeccond}) means that either $E\in \langle \mathcal H_I, \mathcal G_0^{(n)}, iI \rangle$ or $E\notin \mathcal Z(\mathcal H)$. Further, one can verify that $E\otimes I$ not belonging to the second set of Eq.~(\ref{stabformoaqeccond}) is equivalent to $E$ not belonging to the union $\bigcup_{i \neq j}  T_i^{(n)} (T_j^{(n)})^{-1}  \mathcal Z(\mathcal H)$, where this makes use of the form of transversal operators $T_i = T_i^{(n)} \otimes I$ and the construction of $\mathcal S$ from $\mathcal H$. Combining these constraints yields the intersection of sets given in Eq.~(\ref{noiselessBobeqn}),  and the result follows. 
\end{proof}

As a simple illustration of the theorem, we can apply it to the example considered in the previous section. Note in that case that $\mathcal H_I = \langle Z_3, Z_4 \rangle$ and $\mathcal G_0^{(6)} = \{ X_5, Z_5 \}$, and further in that example we have $\mathcal Z(\mathcal H)=  \langle Z_3, Z_4, X_5, Z_5, X_6, Z_6 \rangle$. For a set of error operators to be correctable, we need their products to belong to the two sets in the intersection of Eq.~(\ref{noiselessBobeqn}). Consideration of membership in the first set is straightforward. And, as $\mathcal T_0 = \{ I, X_3, X_4\}$ in this case, belonging to the second set means not belonging to the three sets $X_3^a X_4^b \mathcal Z(\mathcal H)$ for $a,b=0,1$ and $a,b$ not both zero. One can check, for instance, that any set of error operators from the sets $\mathcal H_I$, $\mathcal G_0^{(6)}$, and $\{ X_1, X_2\}$ satisfy these conditions and hence are correctable. 

For the EAOAQEC framework, Theorem~\ref{noiselesserrorcorrecthm} allows us to introduce a notion of code distance that generalizes the corresponding distances for EAQEC, EAOQEC and EACQ codes. Before doing so, let us note that Eq.~(\ref{noiselessBobeqn}) can be equivalently written with the set intersection in the middle factored through the union of the first two sets, and that the first set in the new union will not change since it is contained in $\mathcal Z(\mathcal H)$ and from the fact that the transversal operators are taken from different normalizer cosets on the extended space.  This leads to a formulation of the relevant operator set in the theorem in a way that is more convenient for defining code distances, as in the following, wherein the set considered in Eq.~(\ref{noiselessBobeqn_dist}) is the complement of the set in Eq.~(\ref{noiselessBobeqn}). 

\begin{definition}\label{def:distance}
Given an $n$-qubit entanglement-assisted (EAOAQEC) code $C = C(\mathcal H, \mathcal S, \mathcal G_0, \mathcal L_0, \mathcal T_0)$, with $\mathcal T_0 = \{ T_i = T_i^{(n)} \otimes I \}_i$, we define the {\it code distance $d(C)$ of $C$} as:  
    \begin{equation}
        d(C) = \min\mbox{wt} \left(\left(\mathcal Z(\mathcal H)\backslash \langle \mathcal H_I, \mathcal G_0^{(n)}, iI \rangle \right)\bigcup \big( \bigcup_{i \neq j}  T_i^{(n)} (T_j^{(n)})^{-1} \mathcal Z(\mathcal H)\big)\right), \label{noiselessBobeqn_dist}
    \end{equation}
    where $\min \mbox{wt}(\mathcal A)$ is the minimum of the  weight of operators in a set $\mathcal A \subseteq \mathcal P_n$.
\end{definition}

\begin{remark}
 We note that the error correction theorems and the distance definitions of stabilizer codes \cite{gottesman1997stabilizer}, EAQEC codes \cite{brun2014catalytic}, EAOQEC codes \cite{hsieh2007general}, EACQ codes \cite{kremsky2008classical}, OQEC codes \cite{Poulin2005Stabilizer} and OAQEC codes \cite{kribs2023stabilizer} are all special cases of those in \Cref{noiselesserrorcorrecthm} and \Cref{def:distance}. The code distance is also called the minimum distance of the code. The distance defined in \Cref{def:distance} is sometimes called the `dressed distance' in the literature, and this will be the distance notion we will focus on below. But for reference, the `bare distance' of an EAOAQEC code is given by: 
 \[d_{\mathrm{bare}}(C) = 
 \min\mbox{wt} \bigg(\left(\langle \mathcal H_I, \mathcal L_0^{(n)}, iI \rangle\backslash \langle \mathcal H_I,iI \rangle \right) \bigcup \big( \bigcup_{i \neq j}  T_i^{(n)} (T_j^{(n)})^{-1} \mathcal Z(\mathcal H)\big)\bigg) .
 \]
 Also note that \Cref{def:distance} is based on the operators acting on Alice's qubits and does not involve Bob's qubits as we assume Bob's qubits to be maintained error-free. When Bob's qubits are not error-free, the distance is then defined based on the Pauli group defined over Alice's and Bob's qubits; that is, 
 \[
 d_{\mathrm{noisyBob}}(C) = \min\mbox{wt} \left(\left(\mathcal Z(\mathcal S)\backslash \langle \mathcal S, \mathcal G_0, iI \rangle \right)\bigcup \big( \bigcup_{i \neq j}  T_i T_j^{-1} \mathcal Z(\mathcal S)\big)\right),
 \]
 similar to the OAQEC code distance. 
\end{remark}

\section{EAQEC and EAOQEC as Special Cases of EAOAQEC} \label{sec:EAOAQEC_special_cases}

It is readily apparent that the EAOAQEC codes that one obtains with empty gauge set ($\mathcal G_0 =  \emptyset$) and singleton transversal set ($\mathcal T_0 = \{ I \}$) are exactly the EA codes from the original framework of EAQEC \cite{brun2006correcting}. The subsystem codes from the extension of that framework  \cite{hsieh2007general} are also easily seen to be captured in the non-trivial gauge group and singleton transversal set case, as is pointed out below, along with the error correction theorem for that setting. With some more work, in the next section we will also show that the classically enhanced EA (subspace) codes from  \cite{kremsky2008classical} are obtained as a special case. Figure~\ref{fig:EAOAQEC Diagram} illustrates the hierarchy of special cases.

\begin{figure}
    \centering
    \includegraphics[width=0.65\linewidth]{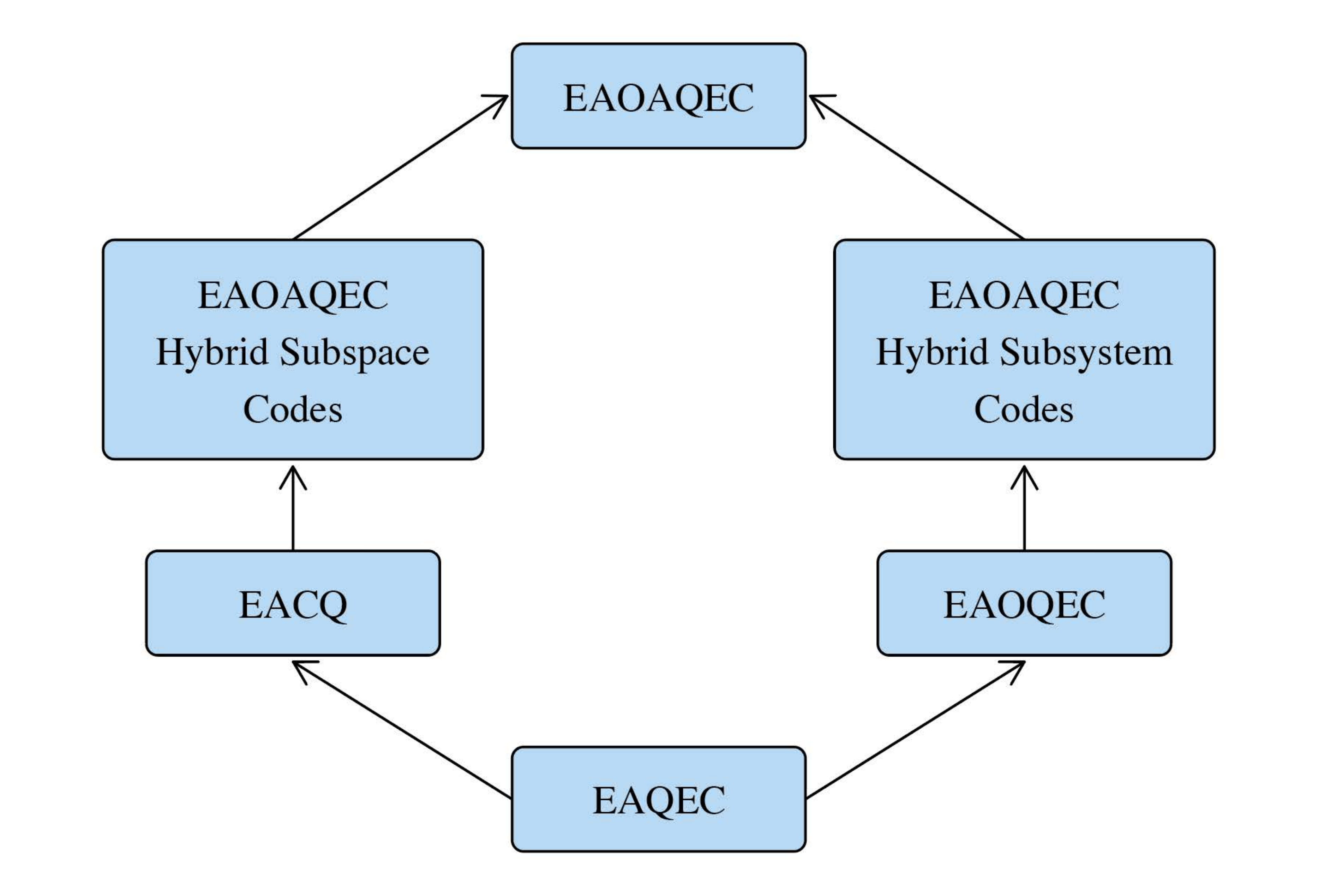}
    \caption{Hierarchy of entanglement-assisted error correction frameworks. Arrows indicate proper inclusions of code types. Conceptually the subspace and subsystem cases are typically treated separately, even though strictly speaking the former is a special case of the latter, and so reflecting this viewpoint we have not included an arrow from the top left box to the top right box. 
    }
    \label{fig:EAOAQEC Diagram}
\end{figure}

\subsection{Subsystem Codes} 

Here we briefly indicate how the general entanglement-assisted subsystem codes and the main error correction result from \cite{hsieh2007general} are obtained as special cases of EAOAQEC. In Section~7 we will return to the subsystem code case (both hybrid and non-hybrid) for further analysis and constructions. 

We will use the notation of \cite{hsieh2007general}, meshing it with ours as needed. To avoid extra notation we will simply focus on the motivating class of codes; all others can be obtained as unitary conjugations of these codes. 
As above, we will have $n= s + e + k + r$, where $s$ is the number of isotropic qubits, $e$ is the number of ebits, and then the subsystem code structure is generated with $k$ logical qubits and $r$ gauge qubits. The starting point is a (in general, non-Abelian) $n$-qubit Pauli subgroup $\mathcal H = \langle \mathcal H_{I} , \mathcal H_{S} \rangle$, with isotropic and symplectic subgroups respectively given by: 
\[
\mathcal H_{I}  =  \langle Z_1 ,  \ldots , Z_{s} \rangle  \quad \quad \mathrm{and} \quad \quad  
\mathcal H_{S}  =  \langle Z_{s+1}, X_{s+1} , \ldots , Z_{s+ e + r}, X_{s+ e + r}  \rangle ,  
\]
and where $\mathcal H_{S}$ is further divided into {\it entanglement} and {\it gauge} subgroups, 
\[
\mathcal H_{E}  =  \langle Z_{s+1}, X_{s+1} ,  \ldots , Z_{s + e}, X_{s+e} \rangle  \quad \mathrm{and} \quad  
\mathcal H_{G}  =  \langle Z_{s+e+1}, X_{s+e+1} , \ldots , Z_{s+ e + r}, X_{s+ e + r}  \rangle . 
\]

We then extend the operators of $\mathcal H$ to $(n+e)$-qubit operators $\widetilde{\mathcal H} = \langle \widetilde{\mathcal H}_I, \widetilde{\mathcal H}_{E} , \widetilde{\mathcal H}_{G}  \rangle$, with $\widetilde{\mathcal H}_I$, $\widetilde{\mathcal H}_G$ extended with identity operators and $\widetilde{\mathcal H}_E$ generators extended in the usual entanglement-assisted way. The Abelian subgroup $\mathcal S  = \langle \widetilde{\mathcal H}_I, \widetilde{\mathcal H}_{E} \rangle$ will have as its stabilizer subspace: 
\[
C = \mathrm{span} \big\{ \ket{0}^{\otimes s} \ket{\Phi}^{\otimes e} \ket{\phi} \ket{\psi} \, \, : \, \, \ket{\phi}\in (\mathbb C^2)^{\otimes r}, \, \ket{\psi}\in (\mathbb C^2)^{\otimes k}  \big\} , 
\]
where note here the entangled pairs are between qubits $(s+1)$ to $(s+e)$ and Bob's ebits. 

This gives an $(r+k)$-qubit subspace with subsystem structure generated by the gauge operators $\widetilde{\mathcal H}_{G}$, which are the logical operators on the first $r$ qubits of the code. The subsystem (`operator') quantum error correction viewpoint of the code is that any two states that differ just in the element $\ket{\phi}$ encode the same quantum information. And, as investigated in the (extensive) subsystem code literature, there is a trade-off between increased error avoidance and simplified decoding procedures that corresponds to the number of logical and gauge qubits within the fixed total $r+k$ along with other properties of a given code.    

In the notation of  \cite{hsieh2007general}, this defines an $[\![ n,k; r, e]\!]$ EAOQECC. For completeness we include the main error correction theorem of \cite{hsieh2007general} and its proof, which can be viewed as the special case of Theorem~\ref{noiselesserrorcorrecthm} (or Eq.~(\ref{stabformoaqeccond})) obtained with trivial transversal set. 

\begin{theorem} \label{eaoqecerrorcorrectionthm}
Given an $[\![n,k ; r , e]\!]$ EAOQEC code as defined above, a set of errors $\{ E_a \otimes I \}$ with $E_a \in \mathcal P_n$ is correctable if and only if for all $a,b$,   
\begin{equation}\label{eaoqecequation}
E_a^\dagger E_b \in \langle \mathcal H_{I} , \mathcal H_{G} , iI \rangle \bigcup \Big(  \mathcal P_n \setminus \mathcal Z( \langle \mathcal H_I , \mathcal H_E   \rangle )  \Big).
\end{equation} 
\end{theorem}

\begin{proof}
Let $E:= E_a^\dagger E_b$. We can prove this either by applying Theorem~\ref{noiselesserrorcorrecthm} or by using the general statement of Eq.~(\ref{stabformoaqeccond}). From Eq.~(\ref{stabformoaqeccond}) for instance, applied in the non-hybrid case when the second set of the union is absent, 
the relevant set of $(n+e)$-qubit operators that $E$ cannot belong to in this case is the set $\mathcal N(\mathcal S) \setminus \mathcal G$ where $\mathcal G = \langle \mathcal S, \widetilde{\mathcal H}_G , iI \rangle = \langle \widetilde{\mathcal H}, iI \rangle$. However, $E\otimes I$ does not belong to $\mathcal N(\mathcal S) = \mathcal Z(\mathcal S)$ if and only if $E\notin \mathcal Z(\langle \mathcal H_I, \mathcal H_E \rangle)$. Further, $E\otimes I \in \mathcal G$ if and only if $E\in \langle \mathcal H_I, \mathcal H_G , iI \rangle$, and hence the result follows. 
\end{proof}

\section{EACQ and Entanglement-Assisted Hybrid Subspace Codes} \label{sec:eacq}

We next focus on the subclass of EAOAQEC codes with a trivial gauge group, which we call {\it entanglement-assisted (EA) hybrid subspace codes}. We will show how the class of codes given in \cite{kremsky2008classical}, presented as a classical enhancement of entanglement-assisted quantum error correcting codes and called EACQ, can be viewed as a subclass of EA hybrid subspace codes. We will also apply Theorem~\ref{eaoqecerrorcorrectionthm} to obtain an error correction result that extends Theorem~1 of \cite{kremsky2008classical} to general EA hybrid subspace codes, and as a consequence we will prove a distance bound for these codes. 

The class of EACQ codes are obtained as follows. As before, we will use the notation of \cite{kremsky2008classical}, meshing it with ours as needed. The parameters $n$, $s$, $k$ (and here $r=0$), and $e$ are as before, with $n = s + e + k$. We introduce an additional parameter here $c$ that will describe the classical bit strings $i \in \{ 0 , 1, \ldots , 2^c -1 \}$ that can be encoded. Each $i$ can be identified, uniquely through its binary expansion, with a sequence of 0's and 1's that we will denote by $\mathbf{x_i}\in ( \mathbb Z_2 )^c$.   It is further assumed that $c = c_1 + 2c_2$ with $s\geq c_1$ and $e\geq c_2$. 

Consider the following $n$-qubit (in general non-Abelian) group $\mathcal S = \langle \mathcal S_C, \mathcal S_Q  \rangle$ generated by two subgroups, which are called the {\it classical stabilizer} $\mathcal S_C = \langle \mathcal S_{C,I} , \mathcal S_{C,S} \rangle$  and {\it quantum stabilizer} $\mathcal S_Q = \langle \mathcal S_{Q,I} , \mathcal S_{Q,S} \rangle$, and are given by: 
\[
\begin{array}{ccl}
\mathcal S_{C,I} & = & \langle Z_1 , Z_2 , \ldots , Z_{c_1} \rangle , \\ 
\mathcal S_{C,S} & = & \langle Z_{s+1}, X_{s+1} , \ldots , Z_{s+ c_2}, X_{s+ c_2}  \rangle , 
\end{array}
\]
and 
\[
\begin{array}{ccl}
\mathcal S_{Q,I} & = & \langle Z_{c_1+1} , Z_{c_1 + 2} , \ldots , Z_{s} \rangle , \\ 
\mathcal S_{Q,S} & = & \langle Z_{s+c_2+1}, X_{s+c_2+1} , \ldots , Z_{s+ e}, X_{s+ e}  \rangle .
\end{array}
\]
As in EAQEC, we can define an Abelian extension of $\mathcal S$ in the usual way, with the extended operators denoted by $\widetilde{S}$, which act on $(n+e)$-qubit  space (note there are $e$ symplectic pairs in $\mathcal S$, with $c_2$ in $\mathcal S_{C,S}$ and $e-c_2$ in $\mathcal S_{Q,S}$), the isotropic operators extended with identity operators, and the symplectic operators extended with $X$'s and $Z$'s as needed. 

The corresponding EACQ code is defined by the following encodings. First, corresponding to $i = 0$, consider the $k$-qubit subspace $C^{(0)}$ of $(n+e)$-qubit space given as follows: 
\[
C^{(0)} = \mathrm{span} \big\{ \ket{\psi}^{(0)} = \ket{0}^{\otimes s} \ket{\Phi}^{\otimes e} \ket{\phi} \, \, : \, \, \ket{\phi}\in (\mathbb C^2)^{\otimes k} \big\}, 
\]
where $\ket{\phi}$ is an arbitrary $k$-qubit state, and $\ket{\Phi} = \frac{1}{\sqrt{2}} ( \ket{00} + \ket{11})$ is the standard maximally entangled state split between each pair from Bob's $e$ qubits and qubits $s+1$ to $s+e$ on the original system. Observe that $C^{(0)}$ is the stabilized subspace of the extended (Abelian) group $\widetilde{S}$. 

Next, given two vectors ${\mathbf x} = (x_1,\ldots , x_n) , \mathbf z = (z_1,\ldots , z_n) \in (\mathbb Z_2)^n$, we define the $n$-qubit operator 
\[
N_{({\mathbf z} | {\mathbf x})} = Z^{z_1} X^{x_1} \otimes  Z^{z_2} X^{x_2} \otimes \ldots  \otimes Z^{z_n} X^{x_n}.
\] 
Then for every other $i$, define a $k$-qubit subspace $C^{(i)}$ spanned by the following $(n+e)$-qubit states, which are obtained by applying specific operators to the states of $C^{(0)}$: 
\begin{equation}\label{transoperators}
\ket{\psi}^{(i)} = \big(  N_{({\mathbf 0} | {\mathbf x}_a)}    \ket{0}^{\otimes c_1} \big)  \ket{0}^{\otimes (s- c_1) } \,  
\big[ \big(  N_{({\mathbf x}_{b_2} | {\mathbf x}_{b_1})} \otimes I \big)   \ket{\Phi}^{\otimes c_2} \big) \big] \,
  \ket{\Phi}^{\otimes (e- c_2)} \, \ket{\phi}, 
\end{equation}
where ${\mathbf x}_a \in (\mathbb Z_2)^{c_1}$ and ${\mathbf x}_{b_1} , {\mathbf x}_{b_2} \in (\mathbb Z_2)^{c_2}$, the $I$ is the identity operator acting on the first $c_2$ of Bob's qubits, and $ {\mathbf x}_i = ( {\mathbf x}_a, {\mathbf x}_{b_1}, {\mathbf x}_{b_2})$ is a splitting of the $c = c_1 +2 c_2$ coordinates of ${\mathbf x}_i$ into its $c_1$, $c_2$, and final $c_2$ coordinate listings. 

For our purposes, we will denote the operators defined by Eq.~(\ref{transoperators}) as $N_i$, so that $\ket{\psi}^{(i)} = N_i \ket{\psi}^{(0)}$. Then, as a point that is relevant below in the error correction related considerations, note that the operators $N_i$ act as the identity operator on the extra ebits. Also observe that these operators form a multiplicatively closed set up to phases. 

As stated in the proof of Theorem~1 from \cite{kremsky2008classical}, an EACQ code is thus defined by a pair of groups $\mathcal H=\langle S_Q^\prime, S_C^\prime \rangle$ and a unitary $U$ such that $S_Q^\prime = U S_Q U^\dagger$ and $S_C^\prime = U S_C U^\dagger$. The codewords are obtained by application of the unitary to the canonical states; $\Psi_i:= U \ket{\psi}^{(i)}$, which is described by the operators $T_i = U N_i U^\dagger$, and so $\Psi_i = T_i \Psi_0$. 

Observe that in Eq.~(\ref{transoperators}) for the canonical EACQ codes we thus have a family of operators $\{ N_i \}$, indexed by $i\in \{ 0, 1, \ldots , 2^c -1\}$, which act on $(n+e)$-qubit space. Notice that $C^{(i)} = N_i C^{(0)}$, and that these subspaces are mutually orthogonal. Further, observe that the sets $N_i \mathcal N(\widetilde{S})$ define different cosets for different $i$; indeed, this follows from the fact that any product $N_i N_j$, with $i \neq j$, does not commute with some element of $\widetilde{S}$ and hence does not belong to $\mathcal Z(\widetilde{S}) = \mathcal N(\widetilde{S})$. It follows, taking everything together, that we have an EAOAQEC code here with a trivial gauge group (i.e., an EA hybrid subspace code), and with the transversal operators $\{N_i \}$ encoding the hybrid classical component of the code. This structure is clearly carried over to general EACQ codes by the unitary conjugations. 

In this sense we can view EACQ codes as forming a subclass of EA hybrid subspace codes. It is clear from the EACQ construction that such codes form a proper subset though; indeed, one simple observation as noted above is that the transversal operators defined in Eq.~(\ref{transoperators}) form a group up to scalar factors, whereas the EAOAQEC construction does not make this demand (only that the transversal operators define different normalizer cosets). In fact there are other constraints as the results that follow make clear.  
To investigate this aspect, we shall say an EA hybrid subspace code is {\it representable as an EACQ code} when an EACQ code can be defined from the group generators as classical and quantum stabilizers that have the same codespace and transversal set as that of the original code. 

The following result identifies a rather restrictive necessary condition that EA hybrid subspace codes must satisfy to be EACQ representable. We give the complete characterization later in this section. 

\begin{lemma}\label{lem:eahybridsubspace}
Suppose that an entanglement-assisted hybrid subspace code defined by a group $\mathcal{H}$ and a coset transversal subset $\mathcal{T}_0$ with elements that act as the identity on the extended space, so $C(\mathcal H, \mathcal S, \mathcal G_0 = \emptyset, \mathcal L_0, \mathcal T_0)$, is representable as an EACQ code. Then the following condition is satisfied:
\begin{equation}\label{eacqsubclass1}
\mathcal{H} \subseteq \mathcal{Z}(Z(\mathcal{Z}(\langle \mathcal{T}_0^{(n)} \rangle) \cap \mathcal{H})), 
\end{equation} 
where $Z(\cdot)$ is the center of the group and $\mathcal{Z}(\cdot)$ is the centralizer of the set inside the Pauli group. Further, if Eq.~(\ref{eacqsubclass1}) is satisfied for a set of coset representatives $\mathcal T_0$, then it is satisfied for any set of representatives from the same cosets that act as the identity on the extended space. 
\end{lemma}

\begin{proof}
If we have an EA hybrid subspace code represented as an EACQ code, to verify Eq.~(\ref{eacqsubclass1}) it is enough to check it for the canonical EACQ codes $\mathcal H = \langle \mathcal S_Q, \mathcal S_C \rangle$ as they define the code class up to unitary conjugations. (Note we use $\mathcal H$ instead of $\mathcal S$ here for the canonical codes as for us $\mathcal S$ is used to denote the extended Abelian group.) For such codes, observe that $S_Q = \mathcal{Z}(\langle \mathcal{T}_0^{(n)} \rangle) \cap \mathcal{H}$, and so from direct calculation with the generators of $\mathcal H$, namely the generators given above for $\{ \mathcal S_{C,I}, \mathcal S_{C,S}, \mathcal S_{Q,I}, \mathcal S_{Q,S} \}$, it follows that 
\[
\mathcal{Z}(Z(\mathcal{Z}(\langle \mathcal{T}_0^{(n)} \rangle ) \cap \mathcal{H})) = \mathcal{Z} ( Z( \mathcal S_Q)) = \mathcal{Z} ( \mathcal S_{Q,I}) \supseteq \langle \mathcal S_C, \mathcal S_Q  \rangle = \mathcal H.
\]

Now we show that if the condition $\mathcal{H} \subseteq \mathcal{Z}(Z(\mathcal{Z}(\langle \mathcal{T}_0^{(n)} \rangle ) \cap \mathcal{H}))$ is satisfied for one subset of coset representatives $\mathcal T_0$, then it is satisfied for any subset of coset representatives $\mathcal T_0^\prime$ that represent the same subset of cosets of $\mathcal N(\mathcal S)$ inside $\mathcal P_{n+e}$. As these coset representatives correspond to the same coset subset, each representative in $\mathcal{T}_0$ differs from the corresponding representative in $\mathcal{T}_0'$ by an element in $\mathcal{N}(S)$. As the subsets of coset representatives for an EAOAQEC code are of the form $M \otimes I$, the elements in $\mathcal{N}(S)$ by which the corresponding representatives differ is of the form $N \otimes I$, where $N$ belongs to $\mathcal{Z}(\mathcal H)$ (as it commutes with the operator over the first $n$ qubits of the stabilizers). 

Now let $L \in \mathcal{Z}(\langle \mathcal{T}_0^{(n)} \rangle) \cap \mathcal{H}$. Let $M'\in \mathcal{T}_0^{\prime (n)}$ and find $M\in \mathcal{T}_0^{(n)}$, $N\in \mathcal Z(\mathcal H)$ such that $M^\prime = MN$. As $N \in \mathcal{Z}(\mathcal H)$ and $L \in \mathcal{H}$, we have $NL = LN$. As $L \in \mathcal{Z}(\langle \mathcal{T}_0^{(n)} \rangle) \cap \mathcal{H}$, we obtain $\mathcal{Z}(\langle \mathcal{T}_0^{(n)} \rangle) \cap \mathcal{H}\subseteq \mathcal{Z}(\langle \mathcal{T}_0^{\prime(n)} \rangle) \cap \mathcal{H}$ since $LM = ML$ implies $LMN = MLN$ which in turn implies $LMN = MNL$ and $LM' = M'L$. Exchanging $\mathcal{T}_0$ and $\mathcal{T}_0'$, we obtain $\mathcal{Z}(\langle \mathcal{T}_0^{\prime (n)} \rangle)\cap \mathcal{H} \subseteq \mathcal{Z}(\langle \mathcal{T}_0^{(n)} \rangle)\cap \mathcal{H}$; hence,  $\mathcal{Z}(\langle \mathcal{T}_0^{(n)} \rangle)\cap \mathcal{H} = \mathcal{Z}(\langle \mathcal{T}_0^{\prime (n)} \rangle)\cap \mathcal{H}$. Thus, if $\mathcal{H} \subseteq \mathcal{Z}(Z(\mathcal{Z}(\langle \mathcal{T}_0^{(n)} \rangle) \cap \mathcal{H}))$ holds true for a choice of $\mathcal{T}_0$, it holds true for any choice of $\mathcal{T}_0'$ which corresponds to the same subset of cosets.
\end{proof}

Taking motivation from Lemma~\ref{lem:eahybridsubspace} and the EACQ construction, in what follows if we are given a group $\mathcal H$ and transversal subset $\mathcal T_0$, we will define the {\it quantum stabilizer subgroup} to be the group $S_Q := \mathcal{Z}(\langle \mathcal{T}_0^{(n)} \rangle) \cap \mathcal{H}$.
We next give an example of an EA hybrid subspace code that is not EACQ representable based on the previous result.  

\begin{example}
Let us consider a 7-qubit and 1-ebit EA hybrid subspace code with the following minimal stabilizer generators and coset transversal subset: $\mathcal H = \langle \mathcal H_I , \mathcal H_S \rangle$ where, 
\begin{align*}   
    \mathcal{H}_S &= \langle S_1:=Z_1Z_2Z_3Z_4X_6X_7, S_2:=Y_3X_4Y_5Z_7 \rangle ,\\
    \mathcal{H}_I &= \langle S_3:=X_3X_4X_5\rangle, 
\end{align*}
and $\mathcal T_0^{(7)} = \langle T_1^{(7)}, T_2^{(7)} \rangle$ where, 
\begin{align*}   
T_1^{(7)} &:= Z_1Z_4Z_7, \\
T_2^{(7)} &:= X_1X_3X_4X_6.
\end{align*}
The following matrix depicts the commutativity of the minimal generators of $\mathcal{T}_0$ with the stabilizers where $0$ and $1$ respectively depict commuting and anti-commuting operators:
\[
    \begin{array}{c||ccc}
    & S_1 & S_2 & S_3\\
\hhline{====}
    T_1^{(7)} & 1 & 1 & 1\\
\hhline{----}
    T_2^{(7)} & 1 & 1 & 0\\
\hhline{====}
    \end{array}
\]

In equation form, this matrix describes the relations $T_i^{(7)} S_j = (-1)^{m_{i, j}} S_j T_i^{(7)}$, where $1\leq i \leq 2$, $1\leq j \leq 3$, and $m_{i, j}$ are the matrix entries.  

We note that the coset transversal operators $T_1^{(7)}$ and $T_2^{(7)}$ have the same commutativity relations with $S_1$ and $S_2$. Thus, $S_1S_2$ will commute with coset transversal operators $T_1^{(7)}$ and $T_2^{(7)}$, and hence is a quantum stabilizer. But $T_1^{(7)}$ anticommutes with both $S_2$ and $S_3$, and so $S_2S_3$ will commute with $T_1^{(7)}$. So we can update the symplectic generators of $\mathcal{H}$ as follows:
\begin{align*}
    \mathcal{H}_S &= \langle S_{12}:=Z_1Z_2X_3Y_4Y_5X_6Y_7, S_{23}:=Z_3Z_5Z_7 \rangle ,\\
    \mathcal{H}_I &= \langle S_3 = X_3X_4X_5\rangle.
\end{align*}

The commutativity relations of the updated minimal generators of $\mathcal{T}_0$ with the stabilizers are depicted in matrix form as follows:
 \[
    \begin{array}{c||ccc}
    & S_{12} & S_{23} & S_3\\
\hhline{====}
    T_1^{(7)} & 0 & 0 & 1\\
\hhline{----}
    T_2^{(7)} & 0 & 1 & 0\\
\hhline{====}
    \end{array}
\]
As $S_{12}$ commutes with all coset transversals, it is a quantum stabilizer. As $S_{23}$ and $S_3$ do not commute with at least one coset transversal, they can be used as classical stabilizers. 
(This idea will be further elucidated in the proof of Theorem~\ref{eahybridsubspacethm} below.)

We note in this case that $Z(\mathcal{Z}(\mathcal{T}_0^{(7)}) \cap \mathcal{H}) = \mathcal{Z}(\mathcal{T}_0^{(7)}) \cap \mathcal{H} = \langle S_{12}\rangle$. As $S_{23} \notin \mathcal{Z}(Z(\mathcal{Z}(\mathcal{T}_0^{(7)}) \cap \mathcal{H}))$, it follows that this code is not EACQ representable. 
\end{example}

In Section~\ref{eacqsubsection} we will establish the full characterization of EACQ as a subclass. The proof is somewhat technical, so let us first consider error correction conditions for these codes without much extra work. We first show how Theorem~\ref{noiselesserrorcorrecthm} applied to EACQ codes reduces to the main error correction theorem from \cite{kremsky2008classical} in that case, and then we prove a general distance bound for EA hybrid subspace codes. 

\begin{theorem} \label{eacqerrorcorrectionthm}
Given an $n$-qubit EACQ code, a set of errors $\{ E_a \otimes I \}$ with $E_a \in \mathcal P_n$ is correctable if and only if for all $a,b$,   
\begin{equation}\label{eacqequation}
E_a^\dagger E_b \in \langle \mathcal S_{Q,I} , \mathcal S_{C,I} \rangle \bigcup \Big(  \mathcal P_n \setminus \mathcal N(\mathcal S_Q)  \Big),
\end{equation} 
where $\mathcal N(\mathcal S_Q) $ is the normalizer of $\mathcal S_Q$ inside $\mathcal P_n$. 
\end{theorem}

\begin{proof}
We can prove this as a direct consequence of Theorem~\ref{noiselesserrorcorrecthm} applied to the canonical EACQ codes given above (which recall define arbitrary EACQ codes up to unitary conjugation of the classical and quantum stabilizer groups). As two technical points on the transversal operators, recall they are of the form $T_i = T_i^{(n)} \otimes I$ for some $n$-qubit operators $T_i^{(n)}$, so that they fit into the setting covered by Theorem~\ref{noiselesserrorcorrecthm}, and also observe that the operators $\{T_i^{(n)}\}_i$ form a multiplicatively closed set up to phases in this case. 

One can verify for this class of codes that 
\[
\mathcal N(\mathcal S_Q) = \mathcal Z(\mathcal S_Q) = \bigcup_{i , j}   T_i^{(n)} (T_j^{(n)})^{-1} \mathcal Z(\mathcal H) ,
\]
and that the isotropic subgroup here is generated by $\mathcal S_{Q,I}$ and $\mathcal S_{C,I}$. 
The conditions of Eq.~(\ref{eacqequation}) thus follow from Theorem~\ref{noiselesserrorcorrecthm}, which can be seen by factoring the set intersection in the middle of Eq.~(\ref{noiselessBobeqn}) through the union of the first two sets and then observing that for EACQ codes the isotropic subgroup is already contained in the last set of that equation.  
\end{proof}

Recall the code distance $d(C)$ defined above, which applies to any EAOAQEC code, including all the subclasses. 

\begin{theorem}\label{coro:distance_EAHybridSubspace_S_Q}
    Let $C$ be an EA hybrid subspace code defined by $\mathcal{H}$ and the coset transversal subset $\mathcal{T}_0$. Let $\mathcal S_Q = \mathcal{Z}(\langle \mathcal{T}_0^{(n)} \rangle) \cap \mathcal H$ and let $C_{\mathcal S_Q}$ be an EAQEC code defined by the group $\mathcal S_Q$. Then, $d(C) \geq d(C_{\mathcal S_Q})$.
\end{theorem}
\begin{proof}
     Let $\mathcal S_Q = \mathcal S_Q^{(I)} \cup \mathcal S_Q^{(S)}$ be an isotropic-symplectic subgroup decomposition of $\mathcal S_Q$. From the generalized distance definition \eqref{noiselessBobeqn_dist}, the distance for the code $S_Q$ is given by 
     \begin{equation}
         d(C_{\mathcal{S}_Q}) = \text{min wt}\left(\mathcal{Z}(\mathcal S_Q) \setminus \mathcal \langle \mathcal{S}^{(I)}_Q, iI\rangle \right).
     \end{equation}
     From the definition of $\mathcal{S}_Q$, one has $T_i^{(n)}(T_j^{(n)})^{-1}\mathcal{Z}(\mathcal{H})\subseteq \mathcal{Z}(\mathcal S_Q)$ for all integer $i$ and $j$. Furthemore, since $T_i^{(n)}(T_j^{(n)})^{-1}\notin \mathcal{Z}(\mathcal{H})$ for all $i\neq j$; one can infer that $T_i^{(n)}(T_j^{(n)})^{-1}E\notin \mathcal{Z}(\mathcal{H})$ for all $E\in \mathcal{Z}(\mathcal{H})$. The arguments above lead to 
     \begin{equation}
         \bigcup_{i\neq j} T_i^{(n)}(T_j^{(n)})^{-1}\mathcal{Z}(\mathcal{H})\subseteq \mathcal{Z}(\mathcal S_Q)\setminus \langle \mathcal{S}_Q^{(I)}, iI \rangle.
     \end{equation}
     To complete the proof, we need to show that 
     \begin{equation} \label{eq:to_complete}
     \mathcal{Z}(\mathcal H)\setminus \langle \mathcal{H}_I, iI \rangle\subseteq \mathcal{Z}(\mathcal S_Q)\setminus \langle \mathcal{S}_Q^{(I)}, iI \rangle\end{equation}
     Indeed, note that $\mathcal{S}_Q\subseteq \mathcal{H}$ implies that $\mathcal{P}_n\setminus \mathcal{Z}(\mathcal{S_Q})\subseteq \mathcal{P}_n\setminus \mathcal{Z}(\mathcal{H})$ and $\langle \mathcal{S}_Q^{(I)}, iI\rangle \subseteq \langle \mathcal{H}_I, iI\rangle\cup (\mathcal{P}_n\setminus \mathcal{Z}(\mathcal{H}))$;
     altogether, we get 
     \begin{equation*}
        \langle \mathcal{S}_Q^{(I)}, iI\rangle \cup (\mathcal{P}_n\setminus \mathcal{Z}(\mathcal{S_Q}))\subseteq \langle \mathcal{H}_I, iI\rangle \cup (\mathcal{P}_n\setminus \mathcal{Z}(\mathcal{H}));
     \end{equation*}
     which gives Eq. \eqref{eq:to_complete} by set complement. In summary 
     \begin{equation}
        \left(\mathcal Z(\mathcal H)\backslash \langle \mathcal H_I, iI \rangle \right)\bigcup \big( \bigcup_{i \neq j}  T_i^{(n)} (T_j^{(n)})^{-1} \mathcal Z(\mathcal H)\big) \subseteq \mathcal{Z}(\mathcal S_Q) \setminus \mathcal \langle \mathcal{S}^{I}_Q, iI\rangle.
     \end{equation}
     Hence, $d(C_{\mathcal{S}_Q})\leq d(C)$ . %
\end{proof}

This result appears to be new for the class of EACQ codes, and so we state that case explicitly. 

\begin{corollary} \label{coro:distance_EACQ_S_Q}
    Let $C$ be an EACQ code defined by quantum and classical stabilizer groups $\mathcal S_Q$ and $\mathcal S_C$, respectively. Let $C_{\mathcal S_Q}$ be an EAQEC code defined by the group $\mathcal S_Q$. Then, $d(C) \geq d(C_{\mathcal S_Q})$.
\end{corollary}

\subsection{Characterization of EACQ Codes as a Subclass of EA Hybrid Subspace Codes}\label{eacqsubsection}

In Theorem~\ref{eahybridsubspacethm} we build on Lemma~\ref{lem:eahybridsubspace} to identify testable conditions that characterize how EACQ codes sit as a subclass of EAOAQEC codes with trivial gauge group; i.e., as a subclass of the class of EA hybrid subspace codes. To do so we require additional notation and preparatory results, beginning with properties of the transversal sets for such codes. 

For ease of presentation we derive the following result just for canonical EACQ codes, but it is readily seen to hold more generally for EA hybrid subspace codes that are EACQ representable by applying unitary conjugations.
As notational preparation for what follows, given a group $(G, *)$ and a subgroup $A$ of $G$, suppose that $\mathcal{M}$ is a set of left coset representatives for some of the left cosets of $A$ as a subgroup of $G$ (we could also focus on right cosets with the same arguments below). Then we define $[\mathcal{M}]_{A}^{G}$ to be the set of (left) cosets of $A$ in $G$ represented by elements of $\mathcal{M}$; that is, 
\begin{align}
     [\mathcal{M}]_{A}^{G} = \{m * A : m \in \mathcal{M}\}.
\end{align}

\begin{lemma}\label{lem:EACQ_coset_subgroup}
    If $\mathcal H$ is an $n$-qubit Pauli subgroup that defines a canonical EACQ code, then the set of $\mathcal Z(\mathcal{S})$ cosets defined by the coset transversal subset $\mathcal T_0 = \{ N_i \}$  is a subgroup of the quotient group $\mathcal{P}_{n+e}/\mathcal{Z}(\mathcal{S})$.
\end{lemma}
\begin{proof}
    The coset transversal subset $\mathcal{T}_0 = \{N_i \}$, for $i \in \{0,\dots,2^{c_1+2c_2}-1\}$, comprises elements that transform the codeword $\Psi_0$ corresponding to the all-zero classical encoding ${\mathbf x}_0 = \mathbf{0}$ to $\Psi_i$ corresponding to ${\mathbf x}_i \in (\mathbb{Z}_2)^{c_1+2c_2}$.  Recall that the operator $N_i = N_{(0|\mathbf{x}_a)} \otimes I^{\otimes (s-c_1)} \otimes N_{(\mathbf{x}_{b_1}|\mathbf{x}_{b_2})} \otimes I^{\otimes (e-c_2)} \otimes I^{\otimes k}$, and the coset in which $N_i$ lies encodes the classical information ${\mathbf x}_i = ( {\mathbf x}_a, {\mathbf x}_{b_1}, {\mathbf x}_{b_2})$.

To see that the set of cosets $[\mathcal{T}_0]_{\mathcal Z(\mathcal{S})}^{\mathcal P_{n+e}} = \{N_i \mathcal Z(\mathcal{S})   : N_i \in \mathcal{T}_0 \}$ forms a subgroup, note first that either $N_i^{-1} = N_i$ or $N_i^{-1} = -N_i$. So both $N_i$ and its inverse define the same normalizer coset, and hence the set is inverse closed. Similarly, given $i$, $j$, either $N_i N_j$ or $- N_i N_j$ belongs to $\mathcal T_0$, and so the set of cosets is multiplicatively closed as well. Further note that $I\in \mathcal T_0$ as a default assumption. It follows that $[\mathcal{T}_0]_{\mathcal{Z}(\mathcal{S})}^{\mathcal{P}_{n+e}}$ is a subgroup of $\mathcal{P}_{n+e}/\mathcal{Z}(\mathcal{S})$. 
\end{proof}

We next turn to a more detailed analysis of the set identified in Lemma~\ref{lem:eahybridsubspace}. 

\begin{lemma}\label{lem:quotient_group_commutative_relation}
Let $G$ be a subgroup of the Pauli group $\mathcal{P}_{n+e}$ and let $A$ be a subgroup of $G$. Then for any $B_2 \subseteq G$ with $\langle[B_2]_A^G\rangle = G/A$, we have $B_2 \subseteq \mathcal{Z}(Z(A))$ if and only if there exists $B_1 \subseteq G$ with $\langle[B_1]_A^G\rangle = G/A$ such that $B_1 \subseteq \mathcal{Z}(A)$.
\end{lemma}

\begin{proof}
We first prove the forward direction, and so suppose we have such a set $B_2$ for which $B_2 \subseteq \mathcal{Z}(Z(A))$. If in fact $B_2 \subseteq \mathcal{Z}(A)$, then we can take $B_1=B_2$. When $B_2 \not\subseteq \mathcal{Z}(A)$, we perform a symplectic-isotropic decomposition on $A$ and ensure we get the isotropic subgroup to be $Z(A)$, and then the symplectic group to be generated by symplectic pairs $\{(x_i,z_i)\}$. As $B_2 \not\subseteq \mathcal{Z}(A)$ but $B_2 \subseteq \mathcal{Z}(Z(A))$, there exist operators in $B_2$ that do not commute with some operators in $\{(x_i,z_i)\}$. Let $B_x^{(i)} = \{b \in B_2 \, | \, bx_i = -x_ib\}$ and $B_z^{(i)} = \{b \in B_2 \, | \, bz_i = -z_ib\}$. Then one-by-one we go through the sets $B_x^{(i)}$ and $B_z^{(i)}$, multiplying elements (on the left) with $z_i$ and $x_i$ respectively and replacing the corresponding elements of $B_2$, to obtain $B_1$ whose elements commute with all operators in $\{(x_i,z_i)\}$. As a result we obtain $B_1 \subseteq \mathcal{Z}(A)$. Thus, we have obtained a set $B_1$ from $B_2$ that satisfies $\langle[B_1]_A^G\rangle = G/A$ and $B_1 \subseteq \mathcal{Z}(A)$. 

We next prove the backward direction. Without loss of generality, for $g=1,2$, we assume that elements in $[B_g]_A^G$ are independent; namely that $B_1 = \{b_1,b_2,\dots,b_l\}$ and $B_2 = \{b_1',b_2',\dots,b_l'\}$ such that $\langle[B_2]_A^G\rangle =\langle[B_1]_A^G\rangle = G/A$. 
So we will show that if $B_1 \subseteq \mathcal{Z}(A)$, then for any $B_2$ with $\langle [B_2]_A^G\rangle = G/A$, it holds that $B_2 \subseteq \mathcal{Z}(Z(A))$. As $\langle [B_1]_A^G\rangle = \langle [B_2]_A^G\rangle$, we can represent every $b_j'$ in terms of the $b_i$ and elements of $A$; i.e., $b_j' = (\prod_{i=1}^{l}b_i^{m_i})a_j$, where $m_i \in \{0,1\}$ and $a_j \in A$. Thus, for any $z \in Z(A)$, and using the fact that $B_1 \subseteq \mathcal{Z}(A)$, we have 
\[
    zb_j' = z\left(\prod_{i=1}^{l}b_i^{m_i}\right)a_j = \left(\prod_{i=1}^{l}b_i^{m_i}\right)za_j  
    = \left(\prod_{i=1}^{l}b_i^{m_i}\right)a_j z = b_j'z . 
\]
It follows that $B_2 \subseteq \mathcal{Z}(Z(A))$, as claimed. 
\end{proof}

\begin{corollary}\label{coro:quotient_group_generator_group_commutative_relation}
    Let $G$ be a subgroup of the Pauli group $\mathcal{P}_{n+e}$ and let $A$ be a subgroup of $G$. The subgroup $G \subseteq \mathcal{Z}(Z(A))$ if and only if there exists $B \subseteq \mathcal{P}_{n+e}$ with $\langle[B]_A^G\rangle = G/A$ such that $B \subseteq \mathcal{Z}(A)$.
\end{corollary}

\begin{proof}
    Suppose that $G \subseteq \mathcal{Z}(Z(A))$, and let $B' \subseteq \mathcal{P}_{n+e}$ be any set with $\langle[B']_A^G\rangle = G/A$, so that $G = \langle A, B' \rangle$ and $B'$ is also contained in $\mathcal{Z}(Z(A))$. Then Lemma~\ref{lem:quotient_group_commutative_relation} 
    gives the existence of a desired set $B$, proving the forward direction of the statement. 
    Next let $B$ be any set such that $\langle[B]_A^G\rangle = G/A$ and $B \subseteq \mathcal{Z}(A)$. Then $G= \langle A, B\rangle$ is generated by two sets that are both contained in $\mathcal{Z}(Z(A))$, and hence $G$ is as well, proving the backward direction.  
\end{proof}

Now we can prove the full characterization. 

\begin{theorem}\label{eahybridsubspacethm}
An entanglement-assisted hybrid subspace code $C(\mathcal H, \mathcal S, \mathcal G_0 = \emptyset, \mathcal L_0, \mathcal T_0)$ defined by a group $\mathcal{H}$ and a coset transversal subset $\mathcal{T}_0$ is representable as an EACQ code if and only if 
$[\mathcal T_0]_{\mathcal{Z}(\mathcal{S})}^{\mathcal{P}_{n+e}}$ is a subgroup of $\mathcal{P}_{n+e}/\mathcal{Z}(\mathcal{S})$ and the following condition is satisfied:
\begin{equation}\label{eacqsubclass}
\mathcal{H} \subseteq \mathcal{Z}(Z(\mathcal{Z}(\mathcal{T}_0^{(n)}) \cap \mathcal{H})), 
\end{equation} 
where recall $\mathcal{T}_0^{(n)}$ is the restriction of operators in $\mathcal{T}_0$ to the first $n$ qubits of the code, $Z(\cdot)$ is the center of the group, and $\mathcal{Z}(\cdot)$ is the centralizer of the set inside the Pauli group. Further, if Eq.~(\ref{eacqsubclass}) is satisfied for a set of coset representatives $\mathcal T_0$, then it is satisfied for any set of representatives from the same cosets. 
\end{theorem}

\begin{proof}
We consider an entanglement-assisted hybrid subspace code defined by $\mathcal{H}$ and a coset transversal subset $\mathcal{T}_0$. Then the forward direction of the proof and the last statement of the theorem follow from Lemma~\ref{lem:eahybridsubspace} and Lemma~\ref{lem:EACQ_coset_subgroup}.

For the backward direction of the proof, we assume the group condition on $\mathcal T_0$ and Eq.~(\ref{eacqsubclass}) are both satisfied. 
Note that all elements in $\mathcal{H}$ that commute with all elements of $\mathcal{T}_0^{(n)}$ forms a subgroup of $\mathcal H$, which we denote by $S_Q = \mathcal{Z}(\mathcal{T}_0^{(n)}) \cap \mathcal{H}$. As above, this will be called the quantum stabilizer group for the code. 

In Algorithm~\ref{alg:SQg_From_H}, starting with a set of independent generators for $\mathcal H$ and the set $\mathcal{T}_0^{(n)}$, we design a procedure to computationally obtain two subsets of the generating set for $\mathcal H$ and an independent subset of the $n$-qubit transversal operators, denoted by $S_{Q_{\text{gen}}}$ (which generate $S_Q$ as a group), $S_{C_{\text{gen}}}$ and $\mathcal T^{(n)}_{C_{\text{gen}}}$ (which generate the group $\langle \mathcal T_0^{(n)} \rangle$) that satisfy the following: 
\begin{enumerate}
    \item $\langle S_{Q_{\text{gen}}} , S_{C_{\text{gen}}} \rangle = \mathcal H$ and $\langle S_{Q_{\text{gen}}} \rangle \cap \langle S_{C_{\text{gen}}} \rangle = \{ I \}$.
    \item Elements of $S_{Q_{\text{gen}}}$ commute with elements of $\mathcal T^{(n)}_{C_{\text{gen}}}$.
    \item Each element of $S_{C_{\text{gen}}}$ anti-commutes with exactly one element of $\mathcal T^{(n)}_{C_{\text{gen}}}$.
\end{enumerate}
We note that the algorithm uses the group structure assumed for the transversal operators in this direction of the proof. 

The one remaining feature we need to identify EACQ structure is to have elements of $S_{Q_{\text{gen}}}$ and $S_{C_{\text{gen}}}$ commute across the sets, which may not be the case for the sets of operators obtained from Algorithm~\ref{alg:SQg_From_H}. This is where Eq.~(\ref{eacqsubclass}) comes in, and allows us to update the set $S_{C_{\text{gen}}}$ while maintaining the above features. Indeed, if we let $S_C$ be the group generated by $S_{C_{\text{gen}}}$ obtained from the algorithm, so $S_C = \langle S_{C_{gen}} \rangle$, then we have $S_C \cap S_Q = S_C \cap (\mathcal{Z}(\mathcal{T}_0^{(n)}) \cap \mathcal{H}) = \{ I\}$. As $\langle S_Q, S_C\rangle = \mathcal{H}$ and $S_C \cap S_Q = \{ I \}$, we can view $S_C$ as representing cosets of $S_Q$ inside $\mathcal{H}$; i.e., $[S_C]_{S_Q}^{\mathcal{H}} = \mathcal{H}/S_Q$. Thus, from Corollary~\ref{coro:quotient_group_generator_group_commutative_relation}, there exists a choice of $S_C$ with $[S_C]_{S_Q}^{\mathcal{H}} = \mathcal{H}/S_Q$ such that $S_C \subset \mathcal{Z}(S_Q)$ if and only if $\mathcal{H} \subset \mathcal{Z}(Z(S_Q)) = \mathcal{Z}(Z(\mathcal{Z}(\mathcal{T}_0^{(n)}) \cap \mathcal{H}))$. Hence, we let $S_C$ be the group generated by this updated set of generators, and call it the classical stabilizer group. 

Now with $S_Q$ and $S_C$ given as above, we can obtain the EACQ representable structure as required. Further, note that an outcome of the algorithm is that $S_{C_{\text{gen}}}$ and $\mathcal T^{(n)}_{C_{\text{gen}}}$ have the same cardinality (the map $A_C$ in the algorithm is a bijection between the sets by construction), and so $|S_{C_{\text{gen}}}| = |T^{(n)}_{C_{\text{gen}}}| = \log | \langle \mathcal{T}_0^{(n)} \rangle |$, which in particular means $|S_C| = |\mathcal T_0^{(n)}|$, and $|S_C|$ describes the classical data that can be encoded. 
\end{proof}

Building on the above proof, we can computationally check whether an EA hybrid subspace code is representable as an EACQ code by performing the following steps: (a) we obtain the quantum and classical stabilizer generators from Algorithm~\ref{alg:SQg_From_H}, (b) we perform the symplectic-isotropic decomposition of the quantum stabilizer generators to obtain the isotropic quantum stabilizer generators $S_{Q_{\text{gen}}}^{(I)}$, and (c) we check whether all classical stabilizer generators commute with all elements in $S_{Q_{\text{gen}}}^{(I)}$ to ensure $\mathcal{H} \subseteq \mathcal{Z}(Z(S_Q)) = \mathcal{Z}(Z(\mathcal{Z}(\mathcal{T}_0^{(n)}) \cap \mathcal{H}))$. Algorithm~\ref{alg:Check_EAOQEC_Is_EACQ} provides an algorithm to check if an EA hybrid subspace code is representable as an EACQ code when $[\mathcal{T}_0]^{\mathcal{P}_{n+e}}_{\mathcal{Z}(\mathcal{S})}$ is known to be a group.

\begin{algorithm}[ht]
\caption{Obtain classical and quantum stabilizer generators}\label{alg:SQg_From_H}
\begin{algorithmic}
\Require $\mathcal{H} = \langle H_1,\dots,H_m\rangle$ and $\mathcal{T}_0^{(n)} = \{ T_1,\dots,T_{c_b} \}$ 
\Ensure Quantum stabilizer generators $S_{Q_{\text{gen}}}$ and classical stabilizer generators $S_{C_{\text{gen}}}$
\Initialize{
$S_{Q_{\text{gen}}} = \emptyset$ \Comment Quantum stabilizer generators\\
$S_{C_{\text{gen}}} = \emptyset$ \Comment Classical stabilizer generators\\
$\mathcal{T}_{C_{\text{gen}}}^{(n)} = \emptyset$\\
$A_C: \mathcal{T}_{C_{\text{gen}}}^{(n)} \rightarrow S_{C_{\text{gen}}}$ \Comment Transversal to classical stabilizer map \\
$\mathcal{T}_{\text{gen}}^{(n)}$ = Generators of $\langle\mathcal{T}_0^{(n)}\rangle$
}
\AlgoProcedure{\For{$i = 1$ to $m$} \Comment Examine if $H_i$ is a quantum or classical stabilizer
    \For{$T_j \in \mathcal{T}_{\text{gen}}^{(n)}$} \Comment Examine if $H_i$ anticommutes with any element in $\mathcal{T}_{\text{gen}}^{(n)}$
        \If{$T_jH_i = -H_iT_j$}
            \State $\mathcal{T}_{C_{\text{gen}}}^{(n)} \rightarrow \mathcal{T}_{C_{\text{gen}}}^{(n)} \cup \{T_j\}$, $\mathcal{T}_{\text{gen}}^{(n)} \rightarrow \mathcal{T}_{\text{gen}}^{(n)} \setminus \{T_j\}$, $A_C(T_j) = H_i$
            \For{$T_p \in \mathcal{T}_{\text{gen}}^{(n)}$}
                    \State $T_p \rightarrow T_pT_j^{\mathbbm{1}_{T_pH_i = -H_iT_p}}$
            \EndFor
            \State $S_{C_{\text{gen}}} \rightarrow S_{C_{\text{gen}}} \cup \{H_i\}$
            \State \textbf{continue loop i}
        \EndIf
    \EndFor
    \For{$T_j \in \mathcal{T}_{C_{\text{gen}}}^{(n)}$} \Comment Examine if $H_i$ anticommutes with any element in $\mathcal{T}_{C_{\text{gen}}}^{(n)}$
        \If{$T_jH_i = -H_iT_j$}
            \State $H_i \rightarrow H_iA_C(T_j)$
        \EndIf
    \EndFor
    \State $S_{Q_{\text{gen}}} \rightarrow S_{Q_{\text{gen}}} \cup \{H_i\}$
\EndFor\\
\Return $S_{Q_{\text{gen}}}$ and $S_{C_{\text{gen}}}$}
\end{algorithmic}
\end{algorithm}

\begin{algorithm}[ht]
\caption{Check if an EA hybrid subspace code is EACQ representable given $[\mathcal T_0]_{\mathcal{Z}(\mathcal{S})}^{\mathcal{P}_{n+e}}$ is a group}\label{alg:Check_EAOQEC_Is_EACQ}
\begin{algorithmic}
\Require $\mathcal{H} = \langle H_1,\dots,H_m\rangle$ and $\mathcal{T}_0^{(n)} = \langle T_1,\dots,T_{c_b}\rangle$ 
\Ensure EA hybrid subspace code is EACQ representable or not
\AlgoProcedure{
\begin{enumerate}
\item Perform Algorithm \ref{alg:SQg_From_H} to obtain the quantum stabilizer generators $S_{Q_{\text{gen}}}$ and the classical stabilizer generators $S_{C_{\text{gen}}}$.
\item Perform symplectic-isotropic decomposition on $S_{Q_{\text{gen}}}$ to obtain $Z(S_Q)$, the isotropic group.
\item If all elements of $S_{C_{\text{gen}}}$ commute with all elements of $Z(S_Q)$, return \textsc{True}; else return \textsc{False}. 
\end{enumerate}}
\end{algorithmic}
\end{algorithm}

\section{\texorpdfstring{Entanglement-Assisted (Hybrid and Non-Hybrid) Subsystem Code\\ Constructions}{Entanglement-Assisted (Hybrid and Non-Hybrid) Subsystem Code Constructions}}\label{sec:eaoaqec_code_construction}

In this section we present a number of EA subsystem code constructions, of both hybrid and non-hybrid varieties. %
The constructions are organized into four subsections. 

Regarding notation, let us briefly recall the parameters of an EAOAQEC code as introduced above: $n$ is the number of physical qubits over which $\mathcal{H}$ is defined; $k$ is the number of logical qubits encoded in the EAOAQEC code; $d$ is the distance of the code; $r$ is the  number of gauge qubits; $e$ is the number of entangled qubits; $m$ is the number of stabilizer generators; $l$ is the total number of symplectic and isotropic qubits; $s$ is the number of isotropic qubits; and $c_b$ is the number of classical bit strings encoded.

Recall for an $n$-qubit subsystem code we have a code subspace $C_{\mathrm{sub}}$ with stabilizer group $\mathcal{S}$, gauge generators $\mathcal{G}_0$, and logical operators $\mathcal{L}_0$. We will denote a minimal set of stabilizer generators of $\mathcal{S}$ by $\{ S_1, S_2, \ldots, S_s \}$.
Let the gauge generators be $\mathcal{G}_0 = \{\overline{G}_{X_1}, \overline{G}_{Z_1}, \cdots, \overline{G}_{X_r}, \overline{G}_{Z_r}\}$, where each $(\overline{G}_{X_i}, \overline{G}_{Z_i})$ is a symplectic pair.
We note that we can always obtain the gauge generators to be in this form \cite{Poulin2005Stabilizer}. 
Further, choose the logical generators to be of the form $\mathcal{L}_0 = \{\overline{L}_{X_1}, \overline{L}_{Z_1}, \cdots, \overline{L}_{X_{n-s-r}}, \overline{L}_{Z_{n-s-r}}\}$, where the number of logical qubits is $n-s-r$. So in our notation, this subsystem code is an $[\![n, n-s-r, d; r]\!]$ code that uses no ebits and encodes only the trivial classical bit string $0$.

The constructions proposed in this section are organized as follows. In \Cref{sec:GaugeToCosetTransversal}, we present the \textit{gauge fixing (GF) construction}, which involves considering some gauge symplectic pairs and converting them to stabilizer-coset transversal pairs, while retaining the rest as gauge generators. In \Cref{sec:isotropictosymplectic}, we provide the \textit{clean qubits (CQ) construction}, where the isotropic stabilizers are converted to extended symplectic operators to obtain an entanglement-assisted code for noiseless ebits. In \Cref{sec:gauge_to_stabilizer}, the \textit{entanglement-assisted gauge fixing (EAGF) construction} presented involves converting pairs of anticommuting gauge generators to stabilizers of the code, which we view as fixing the gauge with entanglement-assistance. The \textit{general gauge fixing (GGF) construction} provided in \Cref{sec:GaugeToCosetTransversal_general} involves a combination of the ideas used in \Cref{sec:GaugeToCosetTransversal} and \Cref{sec:gauge_to_stabilizer}. We present the parameters of the entanglement-assisted hybrid subsystem code obtained from the four proposed constructions in \Cref{tab:construction_summary}.
\begin{table}[ht]
    \centering
\begin{tabular}{|>{\centering\arraybackslash}p{3.7cm}|p{2.4cm}|p{5cm}|p{4.5cm}|}
\hline
  \begin{tabular}{c}
       \textbf{Code}\\\textbf{construction} \\
       \textbf{technique}
  \end{tabular} & \begin{tabular}{c}\textbf{Input}\\\textbf{code}\\ \textbf{parameters}\end{tabular} & ~~~~~~\begin{tabular}{c}\textbf{EAOAQEC}\\ \textbf{code} \\\textbf{parameters}\end{tabular} & ~~~~~~~~\begin{tabular}{c}\textbf{Distance}\\ \textbf{comparison}\end{tabular}\\
\hline
    \begin{tabular}{c}Gauge fixing (GF) \\
    construction\end{tabular} & \!\!\!\!\!\!\begin{tabular}{c} $[\![n,k,d; r,e,c_b]\!]$ \\ with $e\!\geq\! 0, c_b \!\geq\! 1$ \end{tabular} & ~~~\begin{tabular}{c} $[\![n, k, d'; r\!\!-\!y, e,\leq 2^yc_b]\!]$ \\ with $y\leq r$\end{tabular}& \begin{tabular}{c}$d' \geq d$ \\when condition in \\\Cref{prop:dist_gauge_to_cbit} is satisfied.\end{tabular}\\
\hline
  \begin{tabular}{c}
       Clean qubits (CQ)\\
       construction 
  \end{tabular} & \!\!\begin{tabular}{c}
       $[\![n,k,d;r,0,c_b]\!]$\\
       with $c_b \geq 1$
  \end{tabular} &  ~~~~~~~$[\![n-e,k,d';r,e,c_b]\!]$ & ~~~~~~~~~~~$d'\geq d $ \\
  \hline
  \begin{tabular}{c}\!\!\!Entanglement-assisted\\gauge fixing (EAGF)\\ construction\end{tabular} & \!\!\begin{tabular}{c}
       \!\!$[\![n,k,d;r,0,c_b]\!]$\\
       with $c_b \geq 1$
  \end{tabular} &  ~~~~~~~$[\![n,k,d';r-e,e,c_b]\!]$ & ~~~~~~~~~~~$d'\geq d$\\
  \hline\begin{tabular}{c} General gauge\\ fixing (GGF) \\ construction\end{tabular} & \!\!\!\!\!\!\begin{tabular}{c} $[\![n,k,d; r,e,c_b]\!]$ \\ with $e\!\geq\!0, c_b \!\geq\! 1$ \end{tabular} & \!\!\!\!\!\begin{tabular}{c} \!$[\![n, k, d'; r\!\!-\!y_I\!\!-\!y_S\!, \! e\!+\!y_S,\leq 2^{y_I}\!c_b]\!]$ \\ with $y_I + y_S\leq r$\end{tabular}& \begin{tabular}{c}$d' \geq d$ \\ when condition in \\\Cref{coro:dist_gauge_to_cbit_general} is satisfied.\end{tabular}\\
\hline
\end{tabular}
    \caption{Table summarizing the parameters of various EAOAQEC codes constructed from subsystem codes in this section.}
    \label{tab:construction_summary}
\end{table}

We note that some of the examples of EAOAQEC codes in this section are constructed from the $[\![15,1,3;6]\!]$ subsystem color code \cite{Paetznick_Reichardt2013,Bombin2015Gauge} with distance $3$
whose stabilizer and gauge generators are provided in \Cref{app:3D_Subsystem_Color_code} for self-contained reference. For sets of Pauli operators $A$ and $B$, let $AB = \{ab : a\in A, b\in B\}$. In this section, we will denote $\mathcal{A}^* := \mathcal{A}\backslash \{\langle iI\rangle\}$ for any Pauli subgroup $\mathcal{A}$. %

\subsection{Gauge Fixing Construction} \label{sec:GaugeToCosetTransversal}

The first construction of an EAOAQEC code from an entanglement-assisted hybrid subsystem code or its subclasses of codes is based on gauge fixing. Gauge fixing \cite{Poulin2005Stabilizer,Paetznick_Reichardt2013} involves converting a gauge operator to a stabilizer and removing all the gauge operators that anticommute with it from the gauge group. We construct an EAOAQEC code by converting $y$ independent gauge generators to stabilizer generators, where $y \leq r$, and building the coset transversal subset $\mathcal{T}_0$ based on the corresponding anticommuting gauge operators. We note that when $y = r$, the EAOAQEC code obtained would have a trivial gauge group; i.e., it would be a subspace code.  We view the EAOAQEC code to be the initial code whose gauge qubits are fixed based on the encoded classical information. The weight of an operator could be considered as a criteria while choosing operators from a gauge symplectic pair to be added to the stabilizer group or the coset transversal as practically low weight stabilizer generators and high weight coset transversal subset elements are preferred.

We consider $C_{\mathrm{EAOA}} = C(\mathcal{H}, \mathcal{S}, \mathcal{G}_0, \mathcal{L}_0, \mathcal{T}_0)$ to be an $[\![n, k, d;r,e,c_b]\!]$-subsystem code. Without loss of generality, the gauge operators can be written as $\overline{G}_{X_{i}} = \overline{G}_{X_{i}}^{(n)} \otimes I^{\otimes e}$ and $\overline{G}_{Z_{i}} = \overline{G}_{Z_{i}}^{(n)} \otimes I^{\otimes e}$, and we select the $y \leq r$ first symplectic pairs. From these, all the $\overline{G}_{Z_i}$'s are added to the stabilizer generating set, while we define $\mathcal{T}_0'$ to a subset of the group generated by their symplectic partners $\overline{G}_{X_i}$'s to obtain an $[\![n, k, d'; r-y, e, c_b' \leq 2^yc_b]\!]$-hybrid subsystem code $C_{\mathrm{GF}} = C(\mathcal{H}', \mathcal{S}', \mathcal{G}_0', \mathcal{L}_0', \mathcal{T}_0')$ that encodes up to $y$ more classical bits; i.e., it encodes $c_b' \leq 2^yc_b$ classical bit strings. The various sets associated with $C_{\mathrm{GF}}$ are as follows:
 \begin{itemize}
     \item $\mathcal{H}' = \big\langle \mathcal{H}\cup \left\{\overline{G}_{Z_{1}}^{(n)}, \dots, \overline{G}_{Z_{y}}^{(n)}\right\}\big\rangle$,
     \item $\mathcal{S}' = \langle \mathcal{S}\cup \left\{\overline{G}_{Z_{1}}, \dots, \overline{G}_{Z_{y}}\right\}\rangle$,
     \item $\mathcal{L}_0' = \mathcal{L}_0$,
     \item $\mathcal{G}_0' = \{\overline{G}_{X_{y+1}}, \overline{G}_{Z_{y+1}}, \cdots, \overline{G}_{X_{r}},\overline{G}_{Z_{r}}\}$,
     \item $\mathcal{T}_0' \subseteq \mathcal{T}_0G_{X_{\mathcal{T}}}$, where $G_{X_{\mathcal{T}}} := \langle\overline{G}_{X_{1}}, \overline{G}_{X_{2}}, \dots, \overline{G}_{X_{y}}\rangle$.
 \end{itemize}

\begin{lemma}\label{lem:coset_union_eq}
    When $\mathcal{T}_0' = \mathcal{T}_0 G_{X_{\mathcal{T}}}$, the following relation holds:
    \begin{align*}
    \bigcup_{O_i, O_j \in {\mathcal{T}_0'}^{(n)}, O_i \neq O_j}  \!\!\!\!\!\!O_i O_j^{-1} \mathcal Z(\mathcal H')\subseteq \bigcup_{O_i, O_j \in G_{X_{\mathcal{T}}}^{(n)}, O_i \neq O_j}  \!\!\!\!\!O_i O_j^{-1} \mathcal Z(\mathcal H')~~ \bigcup~~\bigcup_{O_i, O_j \in \mathcal{T}_0^{(n)}, O_i\neq O_j}  \!\!\!\!\!O_i O_j^{-1} \mathcal{Z(H)},
    \end{align*}
    with equality satisfied when $[\mathcal{T}_0^{(n)}]_{\mathcal{Z(H)}}^{\mathcal{P}_n}$ is a group.
\end{lemma}
\begin{proof}
We note that
\begin{align}
    &\bigcup_{O_i, O_j \in \mathcal{T}_0^{(n)}G_{X_{\mathcal{T}}}^{(n)}, O_iO_j^{-1} \notin G_{X_{\mathcal{T}}}^{(n)}}  O_i O_j^{-1} \mathcal Z(\mathcal H') \nonumber\\
    &\subseteq \bigcup_{O_i, O_j \in \mathcal{T}_0^{(n)}, O_i\neq O_j}  O_i O_j^{-1} \mathcal \langle \mathcal{Z(H')}, G_{X_{\mathcal{T}}}^{(n)} \rangle = \bigcup_{O_i, O_j \in \mathcal{T}_0^{(n)}, O_i\neq O_j}  O_i O_j^{-1} \mathcal{Z(H)},\text{ as } \mathcal{Z(H)} = \langle \mathcal{Z(H')}, G_{X_{\mathcal{T}}}^{(n)}\rangle,\label{eqn:dist_proof_2}
\end{align}
 where the subset relation follows from the fact that every element in the left hand side set is contained in the right hand side set and, for $O_1,O_2 \in \mathcal{T}_0^{(n)}$ and $G \in G_{X_{\mathcal{T}}}^{(n)}$, the elements of the form $O_1O_2G$ belong to the right hand side set but not to the left hand side set when the representative of the coset corresponding to $O_1O_2$ does not belong to $\mathcal{T}_0^{(n)}$. Thus, the sets are equal when the $[\mathcal{T}_0^{(n)}]_{\mathcal{Z(H)}}^{\mathcal{P}_n}$ is a group.
As $\mathcal{T}_0' = \mathcal{T}_0 G_{X_{\mathcal{T}}}$, we have
    \begin{align}
    &\bigcup_{O_i, O_j \in \mathcal{T}_0^{(n)}G_{X_{\mathcal{T}}}^{(n)}, O_i \neq O_j}  O_i O_j^{-1} \mathcal Z(\mathcal H')\nonumber\\
    &= \bigcup_{O_i, O_j \in G_{X_{\mathcal{T}}}^{(n)}, O_i \neq O_j}  O_i O_j^{-1} \mathcal Z(\mathcal H')~~ \bigcup~~\bigcup_{O_i, O_j \in \mathcal{T}_0^{(n)}G_{X_{\mathcal{T}}}^{(n)}, O_iO_j^{-1} \notin G_{X_{\mathcal{T}}}^{(n)}}  O_i O_j^{-1} \mathcal Z(\mathcal H')\nonumber\\
    &\subseteq \bigcup_{O_i, O_j \in G_{X_{\mathcal{T}}}^{(n)}, O_i \neq O_j}  O_i O_j^{-1} \mathcal Z(\mathcal H')~~ \bigcup~~\bigcup_{O_i, O_j \in \mathcal{T}_0^{(n)}, O_i\neq O_j}  O_i O_j^{-1} \mathcal{Z(H)},\text{ as } \mathcal{Z(H)} = \langle \mathcal{Z(H')}, G_{X_{\mathcal{T}}}^{(n)}\rangle.\label{eqn:dist_proof_2_a}
\end{align}
\end{proof}

 \begin{proposition}\label{prop:dist_gauge_to_cbit}
For a code $C_{\mathrm{GF}}$ constructed from an EAOAQEC code $C_{\mathrm{EAOA}}$ using the gauge fixing construction, when $\mathrm{min~wt}(\mathcal{G'}^{(n)}{G_{X_{\mathcal{T}}}^{(n)}}^*)\geq d(C_{\mathrm{EAOA}})$, then $d(C_{\mathrm{GF}}) \geq d(C_{\mathrm{EAOA}})$, with equality satisfied when ${\mathcal{T}_0'}^{(n)} = \mathcal{T}_0^{(n)}G_{X_{\mathcal{T}}}^{(n)}$ and $[\mathcal{T}_0^{(n)}]_{\mathcal{Z(H)}}^{\mathcal{P}_n}$ is a group, where $\mathcal{G'}^{(n)} =\langle \mathcal{H}_I', \mathcal{G'}_0^{(n)}, iI\rangle$. 
\end{proposition}
\begin{proof}
 By definition, the distance of the codes $C_{\mathrm{EAOA}}$ and $C_{\mathrm{GF}}$ are respectively $d(C_{\mathrm{EAOA}})$ and  $d(C_{\mathrm{GF}})$ as given below:
 \begin{align}
    d(C_{\mathrm{EAOA}}) &= \min\mbox{wt} \left(\left(\mathcal Z(\mathcal H)\backslash \langle \mathcal H_I, \mathcal{G}_0^{(n)}, iI \rangle \right)\bigcup \Bigg( \bigcup_{O_i, O_j \in \mathcal{T}_0^{(n)}, O_i \neq O_j}  O_i O_j^{-1} \mathcal Z(\mathcal H)\Bigg)\right),\label{eqn:dist_c_eaoa}\\
    d(C_{\mathrm{GF}}) &= \min\mbox{wt} \left(\left(\mathcal Z(\mathcal H')\backslash \langle \mathcal H_I', {\mathcal{G}_0'}^{(n)}, iI \rangle \right)\bigcup \Bigg( \bigcup_{O_i, O_j \in {\mathcal{T}_0'}^{(n)}, O_i \neq O_j}  O_i O_j^{-1} \mathcal Z(\mathcal H')\Bigg)\right).\label{eqn:dist_c_eaoa_prime}
 \end{align}  

We focus on the case ${\mathcal{T}'_0}^{(n)} = \mathcal{T}_0^{(n)} G_{X_{\mathcal{T}}}^{(n)}$ as $d(C_{\mathrm{GF}})$ with ${\mathcal{T}'_0}^{(n)} = \mathcal{T}_0^{(n)} G_{X_{\mathcal{T}}}^{(n)}$ is lower bounded by $d(C_{\mathrm{GF}})$ with ${\mathcal{T}'_0}^{(n)} \subseteq \mathcal{T}_0^{(n)} G_{X_{\mathcal{T}}}^{(n)}$. As $\mathcal{Z(H)} = \langle\mathcal{G}^{(n)},\mathcal{L}_0^{(n)}\rangle$ and $\mathcal{G}^{(n)} = \langle\mathcal{H}_I,\mathcal{G}_0^{(n)}, iI \rangle$, we obtain
    \begin{align}
    &\mathcal Z(\mathcal H)\backslash \langle \mathcal{H}_I, \mathcal{G}_0^{(n)}, iI \rangle = \langle \mathcal{H}_{I}, \mathcal{G}_0^{(n)}, \mathcal{L}_0^{(n)}, iI \rangle \setminus \langle \mathcal{H}_{I}, \mathcal{G}_0^{(n)}, iI \rangle = \mathcal{G}^{(n)}\langle\mathcal{L}_0^{(n)}\rangle^*\label{eq_ref1}.
\end{align} 
Similarly,
$\mathcal Z(\mathcal H')\backslash \langle \mathcal{H}_I', {\mathcal{G}_0'}^{(n)}, iI \rangle = \mathcal{G'}^{(n)}\langle\mathcal{L}_0^{(n)}\rangle^*.\nonumber $

For sets $A$, $B$, and $C$ whose elements commute across the sets, when $A$ and $C$ are groups, we obtain 
\begin{align}
    (A B) \cup ( A B C^*) = (A B) \cup ( A C^* B) = (A \cup AC^*) B = \langle A,C\rangle B \label{eqn:set_relation}
\end{align}. 
From \Cref{eqn:set_relation}, we obtain
\begin{align}
    \bigcup_{O_i, O_j \in G_{X_{\mathcal{T}}}^{(n)}, O_i \neq O_j}  O_i O_j^{-1} \mathcal Z(\mathcal H') &= \langle\mathcal{G'}^{(n)},\mathcal{L}_0^{(n)} \rangle{G_{X_{\mathcal{T}}}^{(n)}}^* = \left(\mathcal{G'}^{(n)}{G_{X_{\mathcal{T}}}^{(n)}}^*\right) \cup \left(\mathcal{G'}^{(n)}\langle\mathcal{L}_0^{(n)} \rangle^*{G_{X_{\mathcal{T}}}^{(n)}}^*\right). \label{eq_ref2}
\end{align}
Computing the union of the sets given by \Cref{eq_ref1} and \Cref{eq_ref2}, considering \Cref{eqn:set_relation}, and noting that $\mathcal{G}^{(n)} = \langle\mathcal{G'}^{(n)}, G_{X_{\mathcal{T}}}^{(n)}\rangle$ and $\mathcal{G'}^{(n)}$ and $G_{X_{\mathcal{T}}}^{(n)}$ are groups, we obtain
\begin{align}
    \left(\mathcal{G'}^{(n)}\langle\mathcal{L}_0^{(n)} \rangle^*\right) \cup \left(\mathcal{G'}^{(n)}{G_{X_{\mathcal{T}}}^{(n)}}^*\right)\cup \left(\mathcal{G'}^{(n)}\langle\mathcal{L}_0^{(n)} \rangle^*{G_{X_{\mathcal{T}}}^{(n)}}^*\right) = \left(\mathcal{G}^{(n)}\langle\mathcal{L}_0^{(n)} \rangle^*\right) \cup \left(\mathcal{G'}^{(n)} {G_{X_{\mathcal{T}}}^{(n)}}^*\right),\nonumber\\
      \Rightarrow \mathcal Z(\mathcal H')\backslash \langle \mathcal{H}_I', {\mathcal{G}_0'}^{(n)}, iI \rangle\bigcup\bigcup_{O_i, O_j \in G_{X_{\mathcal{T}}}^{(n)}, O_i \neq O_j}  O_i O_j^{-1} \mathcal Z(\mathcal H') = \left(\mathcal{G}^{(n)}\langle\mathcal{L}_0^{(n)} \rangle^*\right)\cup \left({\mathcal{G}'}^{(n)}{G_{X_{\mathcal{T}}}^{(n)}}^*\right),\label{eqn:dist_proof_11}
\end{align}
which is part of the set for computing $d(C_{\mathrm{GF}})$ in \Cref{eqn:dist_c_eaoa_prime} based on the result in \Cref{lem:coset_union_eq}.

By using \Cref{eqn:dist_proof_11} and \Cref{lem:coset_union_eq}, we obtain
\begin{align}
    &\mathcal Z(\mathcal H')\backslash \langle \mathcal{H}_I', {\mathcal{G}_0'}^{(n)}, iI \rangle ~~\bigcup ~~\bigcup_{O_i, O_j \in \mathcal{T}_0^{(n)} G_{X_{\mathcal{T}}}^{(n)}, O_i \neq O_j}  O_i O_j^{-1} \mathcal Z(\mathcal H')\nonumber\\
    &\subseteq \mathcal Z(\mathcal{H})\backslash \langle \mathcal{H}_I, \mathcal{G}_0^{(n)}, iI \rangle ~~~\bigcup~~ {\mathcal{G}'}^{(n)}{G_{X_{\mathcal{T}}}^{(n)}}^*~~~\bigcup~~\bigcup_{O_i, O_j \in \mathcal{T}_0^{(n)}, O_i\neq O_j}  O_i O_j^{-1} \mathcal{Z(H)},\nonumber\\
    &= \mathcal Z(\mathcal{H})\backslash \langle \mathcal{H}_I, \mathcal{G}_0^{(n)}, iI \rangle ~~~\bigcup~~\bigcup_{O_i, O_j \in \mathcal{T}_0^{(n)}, O_i\neq O_j}  O_i O_j^{-1} \mathcal{Z(H)}~~~\bigcup~~ {\mathcal{G}'}^{(n)}{G_{X_{\mathcal{T}}}^{(n)}}^*.
\end{align}
Therefore, when $\mathrm{min~wt}\left({\mathcal{G}'}^{(n)}{G_{X_{\mathcal{T}}}^{(n)}}^*\right) \geq d(C_{\mathrm{EAOA}})$, we obtain
\begin{align}
     d(C_{\mathrm{GF}}) 
    &\geq  \min\mbox{wt} \left(\mathcal Z(\mathcal{H})\backslash \langle \mathcal{H}_I, \mathcal{G}_0^{(n)}, iI \rangle ~\bigcup~\bigcup_{O_i, O_j \in \mathcal{T}_0^{(n)}, O_i\neq O_j}  O_i O_j^{-1} \mathcal{Z(H)}~\bigcup~ {\mathcal{G}'}^{(n)}{G_{X_{\mathcal{T}}}^{(n)}}^*\right),\nonumber\\
    &= d(C_{\mathrm{EAOA}}).\nonumber
\end{align}
 
\end{proof}
\begin{restatable}{corollary}{CoroGeneral}\label{coro:dist_gauge_to_cbit}
Based on an EAOAQEC code $C_{\mathrm{EAOA}} = (\mathcal{H}, \mathcal{S}, \mathcal{G}_0, \mathcal{L}_0, \mathcal{T}_0)$, let $C_{\mathrm{EAh}}^{(\mathcal{T}_0)} = (\mathcal{H}, \mathcal{S}, \emptyset, \mathcal{G}_0 \cup \mathcal{L}_0, \mathcal{T}_0)$ be the EA hybrid subspace stabilizer code. If $d(C_{\mathrm{EAh}}^{(\mathcal{T}_0)}) = d(C_{\mathrm{EAOA}})$, then $d(C_{\mathrm{GF}}) \geq d(C_{\mathrm{EAOA}})$, with equality satisfied when ${\mathcal{T}_0'}^{(n)} = \mathcal{T}_0^{(n)}G_{X_{\mathcal{T}}}^{(n)}$ and $[\mathcal{T}_0^{(n)}]_{\mathcal{Z(H)}}^{\mathcal{P}_n}$ is a group, where $C_{\mathrm{GF}}$ is obtained from $C_{\mathrm{EAOA}}$ by converting $y$ gauge symplectic pairs to stabilizer-coset transversal pairs.
\end{restatable}
\begin{proof}
 Let $C_{\mathrm{EAh}}^{(\mathcal{T}_0)}$ be the EA hybrid subspace stabilizer code based on the stabilizers of $C_{\mathrm{EAOA}}$ and $\mathcal{T}_0$. The distance of $C_{\mathrm{EAh}}^{(\mathcal{T}_0)}$ is
 \begin{align}
     d(C_{\mathrm{EAh}}^{(\mathcal{T}_0)}) = \min\mbox{wt} \Bigg(\left(\mathcal Z(\mathcal H)\backslash \langle \mathcal{H}_I, iI \rangle \right) \bigcup \Bigg( \bigcup_{O_i, O_j \in \mathcal{T}_0^{(n)}, O_i \neq O_j}  O_i O_j^{-1} \mathcal Z(\mathcal H)\Bigg)\Bigg)
 \end{align}
 as $C_{\mathrm{EAh}}^{(\mathcal{T}_0)}$ has no non-trivial gauge operators. As $\mathcal{G'}^{(n)}{G_{X_{\mathcal{T}}}^{(n)}}^* \subseteq \mathcal Z(\mathcal H)\backslash \langle \mathcal{H}_I, iI \rangle$ and $d(C_{\mathrm{EAOA}}) = d(C_{\mathrm{EAh}}^{(\mathcal{T}_0)})$, we have  \text{min wt}($\mathcal{G'}^{(n)}{G_{X_{\mathcal{T}}}^{(n)}}^*) \geq d(C_{\mathrm{EAOA}})$. 
 Therefore, from \Cref{prop:dist_gauge_to_cbit}, the result follows.
\end{proof}

\begin{example}[15-qubit subsystem color code]
For the $[\![15,1,3;6]\!]$-subsystem color code, when 
we convert the gauge generators $\overline{G}_{X_1}$ and $\overline{G}_{X_2}$ to stabilizers in $\mathcal{S}$ and form $\mathcal{T}_0 = \{I^{\otimes 15}, \overline{G}_{Z_1}, \overline{G}_{Z_2}, \overline{G}_{Z_1}\overline{G}_{Z_2}\}$, the hybrid subsystem code $C_{\mathrm{GF}}$ obtained is an $[\![15,1,3;4,0,4]\!]$ code that encodes $2$ classical bits. We note that $C_{\mathrm{EAOA}}'$ contains $2$ less gauge qubits and encodes $2$ classical bits. 
\end{example}

\begin{example}[15-qubit subsystem hybrid color code]
Let $C_{\mathrm{EAOA}}^{(\mathcal{T}_0)}$ be the $[\![15, 1, 2; 6,0,3]\!]$-hybrid subsystem code obtained from the subsystem code with $\mathcal{T}_0 = \{I^{\otimes 15}, X_5Z_6, X_9Z_{11}\}$ as given in the table below: 
\[
\begin{array}{c|ccccccccccccccc}
\hhline{================}
T_0 & I & I & I & I & I & I & I & I & I & I & I & I & I & I & I\\
T_1 & I & I & I & I & X & Z & I & I & I & I & I & I & I & I & I\\
T_2 & I & I & I & I & I & I & I & I & X & I & Z & I & I & I & I\\
\hhline{================}
\end{array}
\]
 By converting the gauge generators $\overline{G}_{X_1}$ and $\overline{G}_{X_2}$ of $C_{\mathrm{EAOA}}^{(\mathcal{T}_0)}$ to generators of $\mathcal{S}$ and considering $\mathcal{T}_0 = \{I^{\otimes 15}, \overline{G}_{Z_1}, \overline{G}_{Z_2}, \overline{G}_{Z_1}\overline{G}_{Z_2}\} \{I^{\otimes 15}, X_5Z_6, X_9Z_{11}\}$, we obtain a $[\![15, 1, 2; 4,0,12]\!]$-hybrid subsystem code. 
The code encodes $12$ classical bit strings, which we view as encoding $2$ more classical bits as $\mathrm{log}_2(12) = \mathrm{log}_2(3) + \mathrm{log}_2(4) $ $= \mathrm{log}_2(3)$ + 2.
\end{example}

\subsection{Clean Qubits Construction}\label{sec:isotropictosymplectic}

The second construction of an EAOAQEC code from a subsystem code involves first constructing a hybrid subsystem code from a subsystem code using the procedure in \Cref{sec:GaugeToCosetTransversal} and ideas similar to those in \cite{lai2012entanglement} to construct an EAQEC code with imperfect ebits from a stabilizer code. The construction we propose encodes both classical and quantum information and also assumes Bob's qubits to be error-free, unlike \cite{lai2012entanglement}.

Let $C_{\mathrm{hsub}} = \mathcal{C} (\mathcal{S}, \mathcal{G}_0, \mathcal{L}_0, \mathcal{T}_0)$ be an $[\![n, n-s-r, d; r,0,c_b]\!]$ hybrid subsystem code that encodes $c_b = |\mathcal{T}_0|$ classical bit strings and $s$ is the number of stabilizer generators, where $C_{\mathrm{hsub}}$ could be obtained from \Cref{sec:GaugeToCosetTransversal}. 
We obtain the new EAOAQEC code by carefully selecting some qubits as ebits and considering some isotropic operators of $C_{\mathrm{hsub}}$ as the extended symplectic operators and the rest as extended isotropic operators of the code. The main idea is to obtain a stabilizer generating set for $\mathcal{S}$ and a set of qubits $E_Q$ such that for each qubit in $E_Q$ within the generating set there exists exactly one stabilizer generator performing $X$ and exactly one stabilizer generator performing $Z$ on the qubit, with these two stabilizer generators being different. %

Without loss of generality, if we consider $E_Q$ to be the first $|E_Q|$ qubits, then the stabilizer generating set obtained has the following form:
\begin{align} 
    S_1 &= X \otimes I^{\otimes (|E_Q|-1)} \otimes O_1,\nonumber\\
    S_2 &= Z \otimes I^{\otimes (|E_Q|-1)} \otimes O_2,\nonumber\\
    S_3 &= I \otimes X \otimes I^{\otimes (|E_Q|-2)} \otimes O_3,\nonumber\\
    S_4 &= I \otimes Z \otimes I^{\otimes (|E_Q|-2)} \otimes O_4,\nonumber\\
    \vdots\nonumber\\
    S_{2|E_Q|-1} &= I^{\otimes (|E_Q|-1)} \otimes X \otimes O_{2|E_Q|-1},\nonumber\\        S_{2|E_Q|} &= I^{\otimes (|E_Q|-1)} \otimes Z \otimes O_{2|E_Q|},\nonumber\\
    S_i &= I^{\otimes |E_Q|} \otimes O_{i},\text{ for }2|E_Q| < i\leq s. \label{eqn:EAQEC_Stab_form}
\end{align}
From the above form, the restriction of the operators over the last $(n-|E_Q|)$ qubits form symplectic pairs $(O_1,O_2)$, $\cdots$, $(O_{2|E_Q|-1},O_{2|E_Q|})$ and isotropic operators $O_i$, with $2|E_Q|<i\leq s$, and denote the group they generate by $\mathcal{H}$. Thus, we view the stabilizer generators $\{S_j\}_j$ to be the extended operators obtained from $\mathcal{H}$ whose generators are $\{O_j\}_j$. We consider the first $|E_Q|$ qubits to be ebits and the rest to be Alice's qubits. 

We next provide a mathematical condition for the existence of stabilizer generators $S_j$ in the form in \Cref{eqn:EAQEC_Stab_form}.
We also provide a procedure to obtain the stabilizer generating set $\{S_j\}_j$ when $|E_Q| \leq \mathrm{min~wt}(\mathcal{Z}(\mathcal{S})^*)$. 
We define $\mathrm{supp}(O)$ to be the support of an operator $O$ in the Pauli group; i.e., the qubits over which $O$ acts non-trivially. %

\begin{lemma} \label{lem:col_ind_subset_centralizer}
    Suppose there exists a set of stabilizer generators for $\mathcal{S}$ with the form in \Cref{eqn:EAQEC_Stab_form} up to qubit permutations. For any Pauli operator $O \notin <iI>$, if $\mathrm{supp}(O)\subseteq E_Q$, then $O \notin \mathcal{Z(S)}$. 
\end{lemma}
\begin{proof}
    We consider a Pauli operator $O \notin <iI>$ whose support $\mathrm{supp}(O)$ is a subset of $E_Q$. Without loss of generality, let us consider $E_Q$ to be the first $|E_Q|$ qubit indices. Then, $O = A \otimes I^{\otimes (n-|E_Q|)}$ for some non-zero weight $|E_Q|$-qubit operator $A$ and $O$ anticommutes with at least one of the stabilizers as each index in $E_Q$ has a stabilizer with either $X$ and $Z$ on that index and $I$ on rest of the indices of $E_Q$. Thus, $O \notin \mathcal{Z(S)}$. 
\end{proof}
 Each Pauli operator $O$ can be written as $\prod_{i=1}^n X_i^{a_i} Z_i^{b_i}$ upto an overall unimportant phase, and can be represented by its \textit{symplectic representation} $[\overline{a}|\overline{b}]$, where $\overline{a} = [a_i]$ and $\overline{b} = [b_i]$. 
Let $H = [H_1|H_2]$ be the check matrix of a code, which is obtained by stacking the symplectic representation of the stabilizer generators. The syndrome of an operator $O$ with symplectic representation $[o_x|o_z]$ with respect to the code is $H_1o_z^T + H_2o_x^T$. %

\begin{lemma} \label{lem:lin_ind_col_E_Q}
    A set of qubit indices $E_Q$ satisfy $\{Z \mid Z \in \mathcal{Z(S)}^*, \mathrm{supp}(Z) \, \subseteq E_Q\} = \emptyset$ if and only if the columns of the check matrix $H$ corresponding to $E_Q$ are linearly independent. %
\end{lemma}
\begin{proof}
    Let the columns of the check matrix $H$ corresponding to some qubit indices $E_Q$ be linearly independent. Then, there does not exist any non-zero vector $[o_x|o_z]$ with support only on the columns corresponding to qubits in $E_Q$ such that $H_1o_z^T + H_2o_x^T =0$; i.e., has zero syndrome. Thus, there is no operator $O \notin \langle iI\rangle$ with symplectic representation $[o_x|o_z]$ whose support is a subset of $E_Q$ and has zero syndrome. As every element in $\mathcal{Z(S)}$ has zero syndrome, $\{Z\mid Z \in \mathcal{Z(S)}^*,\, \mathrm{supp}(Z) \subseteq E_Q\} = \emptyset$.
    
    We prove the converse by contradiction. Let $\{Z\mid Z \in \mathcal{Z(S)}^*,\, \mathrm{supp}(Z) \subseteq E_Q\} = \emptyset$. We assume that a subset of columns that correspond to qubits in $E_Q$ are linearly dependent. Let $O\notin \langle iI\rangle$ be an operator with a symplectic representation $[o_x|o_z]$ with the elements of $o_x$ and $o_z$ corresponding to the qubits in $E_Q$ being the coefficients in the linear dependency relation between the columns in $E_Q$, and the rest being $0$. Then, $H_1o_z^T + H_2o_x^T = 0$; hence, $O$ has zero syndrome and belongs to $\mathcal{Z(S)}^*$. This is a contradiction as $\{Z\mid Z \in \mathcal{Z(S)}^*, \, \mathrm{supp}(Z) \subseteq E_Q\} = \emptyset$.
\end{proof}

\begin{theorem} \label{prop:EAOAQEC_code_condition_const2}
    Let $C_{\mathrm{hsub}}\left(\mathcal{S}, \mathcal{G}_0, \mathcal{L}_0,\mathcal{T}_0 \right)$ be an $[\![n, k, d; r,0,|\mathcal{T}_0|]\!]$ hybrid subsystem code. An $[\![n-e, k, d'\geq d; r,e,|\mathcal{T}_0|]\!]$ EAOAQEC code $C_{\mathrm{CQ}}\left(\mathcal{H}, \mathcal{S}, \mathcal{G}_0', \mathcal{L}_0',\mathcal{T}_0' \right)$ with stabilizer generators of the form in \Cref{eqn:EAQEC_Stab_form} can be constructed if and only if there exists a subset $E_Q$ of qubits of size $e$ such that no element of $\mathcal Z(\mathcal S)^*$ has its support contained exclusively in $E_Q$.
\end{theorem}
\begin{proof}
    From \Cref{lem:col_ind_subset_centralizer}, we note that if we choose a set $E_Q$ of qubit indices and construct an EAOAQEC code by obtaining a new set of stabilizer generators $S_j$ with the form in \Cref{eqn:EAQEC_Stab_form}, then $\{Z\mid Z \in \mathcal{Z(S)}^*,\,\mathrm{supp}(Z)\subseteq E_Q\} = \emptyset$ as $\mathrm{supp}(Z) \nsubseteq E_Q$ when $Z \in \mathcal{Z(S)}^*$, proving the forward direction.

    We next prove the converse. Suppose there exists a set $E_Q$ of qubit indices such that $\{Z\mid Z \in \mathcal{Z(S)}^*,\, \mathrm{supp}(Z)$ $ \subseteq E_Q\} = \emptyset$. From \Cref{lem:lin_ind_col_E_Q}, the columns of the check matrix $H$ corresponding to the qubits in $E_Q$ are linearly independent. Let $e = |E_Q|$. We next show that we can obtain an $[\![n-e, k, d' \geq d; r,e,|\mathcal{T}_0|]\!]$ EAOAQEC code $C_{\mathrm{CQ}}\left(\mathcal{H},\mathcal{S}, \mathcal{G}_0', \mathcal{L}_0',\mathcal{T}_0'\right)$ with stabilizer generators of the form in \Cref{eqn:EAQEC_Stab_form} from an $[\![n, k, d' \geq d; r,0,|\mathcal{T}_0|]\!]$ hybrid subsystem code $C_{\mathrm{hsub}}\left(\mathcal{H},\mathcal{S}, \mathcal{G}_0, \mathcal{L}_0,\mathcal{T}_0 \right)$. As the columns corresponding to $E_Q$ are linearly independent, we perform Gaussian elimination on $H$ considering $H$ to be an augmented matrix with the columns of the check matrix corresponding to indices in $E_Q$ to be the pivot columns. Thus, only some rows that correspond to the pivot rows based on the pivot columns are row reduced. The stabilizer generators obtained after the Gaussian elimination procedure has the form in \Cref{eqn:EAQEC_Stab_form}, up to qubit permutations. Thus, $2e$ isotropic operators are converted to extended symplectic operators and the rest of the isotropic operators are extended isotropic operators of the EAOAQEC code. The $e$ qubits in $E_Q$ correspond to the ebits; hence, we obtain an $[[n-e, k, d'; r,e,|\mathcal{T}_0|]]$ EAOAQEC code.

    We note that the codespace of the EAOAQEC code and the hybrid subsystem code are the same. The set of errors considered for the two codes vary. The errors on the EAOAQEC code operate only on Alice's qubits, with the assumption that Bob's qubits; i.e., qubits corresponding to indices in $E_Q$, are noise free. For the hybrid subsystem code, errors can have support on all qubits of the code. Thus, the set of errors considered to compute $d'$ compared to $d$ is smaller; hence, $d' \geq d$.
\end{proof}

\begin{corollary}
    From an $[[n, k, d; r,0,|\mathcal{T}_0|]]$ hybrid subsystem code $C_{\mathrm{hsub}}\left(\mathcal{S}, \mathcal{G}_0, \mathcal{L}_0,\mathcal{T}_0 \right)$ with $\text{min~wt}(\mathcal{Z}(\mathcal{S})^*)>1$, for every $e < \text{min~wt}(\mathcal{Z}(\mathcal{S})^*)$, a total of $\binom{n}{e}$ $[\![n-e, k, d'\geq d; r,e,|\mathcal{T}_0|]\!]$ EAOAQEC codes can be constructed.
    \end{corollary}
    
\begin{proof}
    We consider $E_Q$ to be a subset of qubits of $C_{\mathrm{hsub}}$ of size $e$ less than $\mathrm{min~wt}(\mathcal{Z}(\mathcal{S})^*)$. We note that there are $\binom{n}{e}$ ways in which we can obtain $E_Q$. As $e < \mathrm{min~wt}(\mathcal{Z}(\mathcal{S})^*)$ and $e = |E_Q|$, we obtain $\{Z\mid Z \in \mathcal{Z(S)^*},\, \mathrm{supp}(Z) \subseteq E_Q\} = \emptyset$. From \Cref{prop:EAOAQEC_code_condition_const2}, we can construct an $[\![n-e, k, d'\geq d; r,e,|\mathcal{T}_0|]\!]$ EAOAQEC code based on the $[\![n, k, d'\geq d; r,0,|\mathcal{T}_0|]\!]$ hybrid subsystem code. We note that we neglect the permutation of qubits while counting the number of codes that can be constructed.
\end{proof}

We provide the procedure to construct an EAOAQEC code from a hybrid subsystem code using Gaussian elimination in Algorithm \ref{alg:EAOAQEC_HybridSubsystemCode} which is presented in Appendix~\ref{isotosymalgs}.

Calderbank-Shor-Steane (CSS) codes are a class of stabilizer codes for which stabilizer generators can be obtained whose check matrices have the following form: 
\begin{align} \label{eqn:CSS_check_mat}
H = \left[\begin{array}{c|c} H_X & 0\\0 & H_Z\end{array}\right] ;   
\end{align}
i.e., the stabilizer generators can be decomposed into a set of $X$ operators and a set of $Z$ operators. Using the structure of the CSS code's check matrix, we provide a mathematical condition and procedure to construct EAOAQEC code from a hybrid subsystem code when $H_X=H_Z$.

\begin{proposition} \label{prop:EAOA_CSS_code_condition_const2}
    From an $[\![n, k, d; r,0,|\mathcal{T}_0|]\!]$ hybrid subsystem CSS code $C_{\mathrm{hsub}}\left(\mathcal{S}, \mathcal{G}_0, \mathcal{L}_0,\mathcal{T}_0 \right)$ based on a dual-containing classical code; i.e., $H_X=H_Z=H$, an $[\![n-e, k, d'\geq d; r,e,|\mathcal{T}_0|]\!]$ EA operator algebra CSS code $C_{\mathrm{CQ}}\left(\mathcal{H}, \mathcal{S}, \mathcal{G}_0', \mathcal{L}_0',\mathcal{T}_0' \right)$ with stabilizer generators of the form in \Cref{eqn:EAQEC_Stab_form} can be constructed if and only if $e \leq \mathrm{rank}(H) = (n-k)/2$. %
\end{proposition}
\begin{proof}
 From Eq.~(\ref{eqn:CSS_check_mat}), the columns based on $H_X$ are not linearly dependent on the columns based on $H_Z$. Thus, it is sufficient to consider linear independence within $H_X = H_Z = H$. As $H_X$ has at least one set of $\mathrm{rank}(H_X)$ linearly independent columns, for any $e \leq \mathrm{rank}(H_X)$, we can find a set containing $e$ linearly independent columns through Gaussian elimination. We obtain $E_Q$ to be a subset of the pivotal column indices of the Gaussian elimination procedure, with $|E_Q|=e \leq \mathrm{rank}(H_X)$. As the columns of $H_X$ corresponding to the qubits in $E_Q$ are linearly independent, from \Cref{lem:lin_ind_col_E_Q}, we obtain that $\{Z\mid Z \in \mathcal{Z(S)}^*,\, \mathrm{supp}(Z) \subseteq E_Q\} = \emptyset$. From \Cref{prop:EAOAQEC_code_condition_const2}, we can construct an $[\![n-e, k, d'\geq d; r,e,|\mathcal{T}_0|]\!]$ EA operator algebra CSS code $C_{\mathrm{CQ}}\left(\mathcal{H},\mathcal{S}, \mathcal{G}_0', \mathcal{L}_0',\mathcal{T}_0' \right)$ with stabilizer generators of the form in Eq.~(\ref{eqn:EAQEC_Stab_form}) from an $[\![n, k, d'\geq d; r,0,|\mathcal{T}_0|]\!]$ hybrid subsystem CSS code $C_{\mathrm{hsub}}\left(\mathcal{H},\mathcal{S}, \mathcal{G}_0', \mathcal{L}_0',\mathcal{T}_0' \right)$. As there are at most $rank(H)$ columns of $H$ that are linearly independent, from \Cref{lem:col_ind_subset_centralizer} and \Cref{lem:lin_ind_col_E_Q}, there cannot exist $E_Q$ of size $e > rank(H)$ with stabilizer generators of the form in \Cref{eqn:EAQEC_Stab_form}. Hence, when the code is constructed, $e$ should be at most $rank(H)$.
\end{proof}

Using the structure of the CSS code, we modify the code construction procedure provided in \Cref{alg:EAOAQEC_HybridSubsystemCode} to provide a more efficient procedure tailored to CSS codes in \Cref{alg:EAOAQEC_HybridSubsystemCSSCode}, which is also presented in Appendix~\ref{isotosymalgs}.

We next summarize the various sets of EAOAQEC codes constructed here in terms of those of the hybrid subsystem code. Let $O^{(E_Q)}$ and $O^{(E_Q^{\complement})}$ be the restriction of the operator $O$ on the qubits in $E_Q$ and $E_Q^{\complement}$, respectively, where $E_Q^{\complement}$ denotes the complement of set $E_Q$. %
The EAOAQEC code $C_{\mathrm{EAOA}} = C(\mathcal H, \mathcal S, \mathcal G_0', \mathcal L_0', \mathcal T_0')$ is obtained from a hybrid subsystem code $C_{\mathrm{hsub}} = C( \mathcal S, \mathcal G_0, \mathcal L_0, \mathcal T_0)$ with the relevant sets being as follows:
\begin{itemize}
\item $\mathcal H = \langle \{O^{(E_Q^{\complement})} \mid O \in \mathcal S \}\rangle$ %
\item The stabilizer group for both the codes is $\mathcal S$,
\item $\mathcal G_0' =  \{O^{(E_Q^{\complement})} \mid G \in \mathcal{G}_0, O \in G\mathcal{S}, O^{(E_Q)} = I \}$, %
\item $\mathcal L_0' = \{O^{(E_Q^{\complement})} \mid L \in \mathcal{L}_0,O\ \in L\mathcal{S}, O^{(E_Q)} = I \}$,
\item $\mathcal T_0' = \{O^{(E_Q^{\complement})}\mid T \in \mathcal{T}_0, O \in T\mathcal{Z(S)}, O^{(E_Q)}=I \}$. %
\end{itemize}

\begin{example}[15-qubit subsystem color code]
 For the subsystem color code provided in \Cref{app:3D_Subsystem_Color_code}, each stabilizer is of weight at least $8$ and the distance is $3$. Thus, the weight of non-identity elements in $\langle \mathcal{S}, \mathcal{G}_0, \mathcal{L}_0 \rangle^*$ is at least $3$. We choose the first $2$ qubits to convert them to ebits. We consider the check matrix of the subsystem code:
 
$$\left[
 \begin{array}{c|c}
    \begin{array}{ccccccccccccccc}
    \mathbf{1} & \mathbf{0} & 1 & 0 & 1 & 0 & 1 & 0 & 1 & 0 & 1 & 0 & 1 & 0 & 1\\
    \mathbf{1} & \mathbf{1} & 0 & 0 & 1 & 1 & 0 & 0 & 1 & 1 & 0 & 0 & 1 & 1 & 0\\
    \mathbf{0} & \mathbf{0} & 0 & 1 & 1 & 1 & 1 & 0 & 0 & 0 & 0 & 1 & 1 & 1 & 1\\
    \mathbf{0} & \mathbf{0} & 0 & 0 & 0 & 0 & 0 & 1 & 1 & 1 & 1 & 1 & 1 & 1 & 1        
    \end{array} & \bm{0} \\ 
    \bm{0} &\begin{array}{ccccccccccccccc}
    \mathbf{1} & \mathbf{0} & 1 & 0 & 1 & 0 & 1 & 0 & 1 & 0 & 1 & 0 & 1 & 0 & 1\\
    \mathbf{1} & \mathbf{1} & 0 & 0 & 1 & 1 & 0 & 0 & 1 & 1 & 0 & 0 & 1 & 1 & 0\\
    \mathbf{0} & \mathbf{0} & 0 & 1 & 1 & 1 & 1 & 0 & 0 & 0 & 0 & 1 & 1 & 1 & 1\\
    \mathbf{0} & \mathbf{0} & 0 & 0 & 0 & 0 & 0 & 1 & 1 & 1 & 1 & 1 & 1 & 1 & 1        
    \end{array}   
\end{array}
\right]$$

 We note that the 4 columns of the check matrix corresponding to the first two qubits of the code, marked in bold font, is not in reduced row echelon form. We perform row reduction by adding the $1^{\mathrm{st}}$ and the $5^{\mathrm{th}}$ rows to the $2^{\mathrm{nd}}$ and the $6^{\mathrm{th}}$ rows, respectively. This corresponds to multiplying $S_1$ and $S_5$ to $S_2$ and $S_6$, respectively. We provide the modified stabilizer generators in the table below:
 \[
\begin{array}{c|ccccccccccccccc}
\hhline{================}
S_1 & X & I & X & I & X & I & X & I & X & I & X & I & X & I & X\\
S_2' & I & X & X & I & I & X & X & I & I & X & X & I & I & X & X\\
S_3 & I & I & I & X & X & X & X & I & I & I & I & X & X & X & X\\
S_4 & I & I & I & I & I & I & I & X & X & X & X & X & X & X & X\\
S_5 & Z & I & Z & I & Z & I & Z & I & Z & I & Z & I & Z & I & Z\\
S_6' & I & Z & Z & I & I & Z & Z & I & I & Z & Z & I & I & Z & Z\\
S_7 & I & I & I & Z & Z & Z & Z & I & I & I & I & Z & Z & Z & Z\\
S_8 & I & I & I & I & I & I & I & Z & Z & Z & Z & Z & Z & Z & Z\\
\hhline{================}
\end{array}
\]

By considering the first two qubits that are in bold font as the ebits, we obtain an EAOAQEC code. We note that the gauge generators already have identity operators on the first two qubits. When the gauge generators do not have identity operators on the ebits, they can be transformed to identity operators by multiplying them with the corresponding stabilizer generators containing the same operator on the ebit.

Let the coset transversal subset be given by $\mathcal{T}_0 = \{I^{\otimes 15}, X_3Z_3Z_4Z_5, X_3X_4Z_4X_4Z_5\}$, as given in the below table:
\[
\begin{array}{c|ccccccccccccccc}
\hhline{================}
T_0 & I & I & I & I & I & I & I & I & I & I & I & I & I & I & I\\
T_1 & I & I & XZ & Z & Z & I & I & I & I & I & I & I & I & I & I\\
T_2 & I & I & X & XZ & XZ & I & I & I & I & I & I & I & I & I & I\\
\hhline{================}
\end{array}
\]
As a weight 2 operator $Z_1Y_3$ is in the same coset as $Y_3Z_4Z_5$, neglecting the phase, the distance of the hybrid subsystem code is $2$. For the EAOAQEC code, the distance is $2$ and $3$ when ebits are noisy and noiseless, respectively. Thus, the EAOAQEC code is an $[\![13, 1, 3; 6, 2,3]\!]$ code that encodes $3$ classical bit strings. As the subsystem code is based on a classical code with parity check matrix
\[H_C = \left[\begin{array}{ccccccccccccccc}
    1 & 0 & 1 & 0 & 1 & 0 & 1 & 0 & 1 & 0 & 1 & 0 & 1 & 0 & 1\\
    1 & 1 & 0 & 0 & 1 & 1 & 0 & 0 & 1 & 1 & 0 & 0 & 1 & 1 & 0\\
    0 & 0 & 0 & 1 & 1 & 1 & 1 & 0 & 0 & 0 & 0 & 1 & 1 & 1 & 1\\
    0 & 0 & 0 & 0 & 0 & 0 & 0 & 1 & 1 & 1 & 1 & 1 & 1 & 1 & 1        
    \end{array}\right]
    \]
and $\mathrm{rank}(H_C) = 4$, we can also obtain an EAOAQEC code with $4$ ebits.
\end{example}

\subsection{Entanglement-assisted Gauge Fixing Construction}\label{sec:gauge_to_stabilizer}

We next propose a method for constructing entanglement-assisted codes from subsystem codes.
Consider the code $C_{\mathrm{hsub}} = C(\mathcal S, \mathcal G_0, \mathcal L_0, \mathcal T_0)$ with parameters $[\![n,k,d;r,0,c_b]\!]$.
Suppose the operators in $\mathcal G_0$ are in symplectic form.
Construct a subset of $\mathcal G_0^{(\mathcal{H_S})} \subseteq \mathcal G_0$ formed by taking $e \leq r$ of the symplectic pairs. 
We can convert each of these generators into stabilizers by extending them using ebits in exactly the same way as described in \Cref{sec:EAOAQEC}.
This gives a new code $C_{\mathrm{EAGF}} = C(\mathcal H', \mathcal S', \mathcal G_0', \mathcal L_0', \mathcal T_0')$, where we have: 
\begin{itemize}
\item $\mathcal H' $ is the group generated by $\mathcal G_0^{(\mathcal{H_S})}$ and the generators of $\mathcal S$,
\item $\mathcal S'$ is the Abelian extension of $\mathcal H'$,
\item $\mathcal G_0' = \{G_i \otimes I \mid G_i \in \mathcal G_0 \setminus \mathcal G_0^{(\mathcal{H_S})} \}$,
\item $\mathcal{L}_0' = \{L_i \otimes I | L_i \in \mathcal{L}_0\}$,
\item $\mathcal T_0' = \{ T_i \otimes I | T_i \in \mathcal T_0\}$, where $T_i$ are the generators of $\mathcal T_0$,
\end{itemize}
where $I$ acts on the $e$ ebits introduced in the extension. 

Let $\mathcal{G}_0^{(\mathcal{H_S})}$ contain a gauge symplectic pair $(\overline{G}_{X_1},\overline{G}_{Z_1})$ whose elements are converted to stabilizers by extending using 1 ebit, namely $\overline{G}_{X_1} \otimes X$ and $\overline{G}_{Z_1} \otimes Z$. The states stabilized by $\mathcal{S}$, $\overline{G}_{X_1} \otimes X$ and $\overline{G}_{Z_1} \otimes Z$ have the form $(\ket{\psi_0}\ket{0} + \ket{\psi_1}\ket{1})/\sqrt{2}$, where $\ket{\psi_i}$ is a codeword of $C$ and is a $(-1)^i$-eigenstate of $\overline{G}_{Z_1}$,  such that $\overline{G}_{X_1}\ket{\psi_0} = \ket{\psi_1}$. Thus, $\ket{\psi_i}$ is a codeword of $C$ with its first gauge qubit fixed to $\ket{i}$ and the gauge qubit along with the ebit are in the Bell state $(\ket{00}+\ket{11})/\sqrt{2}$. Thus, the construction proposed corresponds to entanglement-assisted gauge fixing as the entanglement with the ebit assists the gauge qubit to be fixed to the mixed state corresponding to one half of the Bell state.

It is clear that the new code has $n' = n$, $r' = r - e$, $e$ ebits, and  $k' = k$. The latter holds since the $e$ symplectic pairs of gauge operators have become symplectic pairs in $\mathcal{H}'$, subsequently $k' = n - e -s- (r-e) = n -s - r=k$. The code distance of the new code is not so obvious as it is not always possible to extend a minimum weight logical operator to act trivially on the new qubits.
\begin{proposition}\label{prop:distance_gauge_to_stab}
   An entanglement-assisted code $C (\mathcal H', \mathcal S', \mathcal G_0', \mathcal L_0', \mathcal T_0')$, constructed as above from a hybrid subsystem code $C(\mathcal S, \mathcal G_0, \mathcal L_0, \mathcal T_0)$, has code distance $d' \geq d$.
\end{proposition}
\begin{proof}
The correctable error set for code $C_{\mathrm{EAGF}}$ is $\{ F_a \}$ such that
\begin{equation}
    F_a^\dag F_b \in {\bf E'} = \Big( \langle \mathcal H_I', \mathcal G_0', iI \rangle \bigcup \big( \mathcal P_n \setminus \mathcal Z(\mathcal H')\big) \Big) \bigcap \Big( \mathcal P_n \setminus \big( \bigcup_{i \neq  j}  T_i T_j^{-1} \mathcal Z(\mathcal H')  \big)  \Big).
\end{equation}
For code $C$, we have $\mathcal{H}_I = \mathcal{H} = \mathcal{S}$ as $\mathcal S$ is Abelian, and so the correctable error set is $\{ E_a \}$ such that 
\begin{equation}
    E_a^\dag E_b \in {\bf E} = \Big( \langle \mathcal S, \mathcal G_0, iI \rangle \bigcup \big( \mathcal P_n \setminus \mathcal Z(\mathcal S)\big) \Big) \bigcap \Big( \mathcal P_n \setminus \big( \bigcup_{i \neq  j}  T_iT_j^{-1} \mathcal Z(\mathcal S)  \big)  \Big).
\end{equation}
To prove the proposition it suffices to show that ${\bf E}$ is a subset of ${\bf E'}$. Proceeding as we did in the proof of Theorem~\ref{coro:distance_EAHybridSubspace_S_Q} and considering $\mathcal{H}_I' = \mathcal{S}$, we have
\begin{equation}
    \Big( \langle \mathcal S, \mathcal G_0, iI \rangle \bigcup \big( \mathcal P_n \setminus \mathcal Z(\mathcal S)\big) \Big)\subseteq \Big( \langle \mathcal H_I', \mathcal G_0', iI \rangle \bigcup \big( \mathcal P_n \setminus \mathcal Z(\mathcal H')\big) \Big)
\end{equation}
and 
\begin{equation}
    \big(\bigcup_{i \neq  j}  T_i T_j^{-1} \mathcal Z(\mathcal H')  \big)\subseteq \big( \bigcup_{i \neq  j}  T_iT_j^{-1} \mathcal Z(\mathcal S)  \big).
\end{equation}
Hence, it becomes clear that ${\bf E}\subseteq {\bf E}'$, and the result follows. 
\end{proof}

\begin{example}[15-qubit subsystem color code]

Let us return to our example of the 15-qubit subsystem color code.
We consider the operators in the table below:
\[
\begin{array}{c|ccccccccccccccc}
\hhline{================}
I & I & I & I & I & I & I & I & I & I & I & I & I & I & I & I\\
T_1 & I & I & I & I & I & I & I & I & I & I & I & I & Z & Z & Z\\
\hhline{================}
\end{array}
\]
Choosing the transversal subset $\mathcal{T}_0 = \{I, T_1\}$, gives us a $[\![15,1,1;6,0,2]\!]$-hybrid subsystem code that encodes two bit strings. 
One can confirm that there is only one uncorrectable error with weight smaller than three, namely 
\begin{equation}
    I T_1 \overline{G}_{Z_1} = I^{\otimes 11} Z I^{\otimes 3}.
\label{eq:t1t2n}
\end{equation}
We next convert the gauge generators $\overline{G}_{X_1}$ and  $\overline{G}_{Z_1}$ to stabilizers in $\mathcal{H}$ and consider the coset transversal subset $\mathcal{T}'_0 = \{I, T_1 \otimes I\}$. This gives us a $[\![15,1,3; 5,1,2]\!]$-hybrid EAOAQEC code requiring one ebit that encodes two bit strings. 
The distance has increased because the error in Eq.~\eqref{eq:t1t2n} is correctable for the transformed code (as $\overline{G}_{Z_1} \notin \mathcal Z(\mathcal H')$).

\end{example}

\subsection{General Gauge Fixing Construction}\label{sec:GaugeToCosetTransversal_general}

We next provide the construction of an EAOAQEC code from an entanglement-assisted hybrid subsystem code or its subclasses combining ideas from \Cref{sec:GaugeToCosetTransversal} and \Cref{sec:gauge_to_stabilizer}. We call this construction general gauge fixing as it is a combination of gauge fixing and entanglement-assisted gauge fixing. This construction involves converting operators from some gauge symplectic pairs to either stabilizer generators and coset transversal operators. When non-commuting gauge operators are added to the stabilizer generating set, additional entanglement-assistance is required. When a particular set of gauge symplectic pairs are provided, the construction proposed in this section is a non-trivial amalgamation of the constructions in \Cref{sec:GaugeToCosetTransversal} and \Cref{sec:gauge_to_stabilizer}.

We consider an $[[n,k,d;r,e,c_b]]$-EAOAQEC code $C_
{\mathrm{EAOA}} = C(\mathcal{H}, \mathcal{S}, \mathcal{G}_0, \mathcal{L}_0, \mathcal{T}_0)$ with gauge symplectic pairs $(\overline{G}_{X_i},\overline{G}_{Z_i})$. Let $y = y_I + y_S \leq r$. We define the following sets:
\begin{align}
    \mathcal{G}_0^{(\mathcal{H}_I)} \coloneqq \{\overline{G}_{Z_1},\cdots,\overline{G}_{Z_{y_I}}\}\text{ and }\mathcal{G}_0^{(\mathcal{H}_S)} \coloneqq \{\overline{G}_{X_{y_I+1}},\overline{G}_{Z_{y_I+1}},\cdots, \overline{G}_{X_{y}},\overline{G}_{Z_{y}}\}.
\end{align} 

We begin by using the construction in \Cref{sec:GaugeToCosetTransversal} based on the first $y_I$ gauge symplectic pairs and further use the construction in \Cref{sec:gauge_to_stabilizer}
on the next $y_S$ pairs to obtain the $[[n,k,d';r-y_I-y_S,e+y_S,\leq 2^{y_I}c_b]]$-EAOAQEC code $C_{\mathrm{GGF}} = C(\mathcal{H'}, \mathcal{S'}, \mathcal{G'}_0, \mathcal{L'}_0, \mathcal{T'}_0)$ that is based on the following sets:
\begin{itemize}
    \item $\mathcal{H}' = \langle \mathcal{H}, \mathcal{G}_0^{(\mathcal{H}_I)}, \mathcal{G}_0^{(\mathcal{H}_S)}\rangle$
    \item $\mathcal{L}_0' = \mathcal{L}_0$
    \item $\mathcal{G}_0' = \{\overline{G}_{X_{y+1}},\overline{G}_{Z_{y+1}},\cdots, \overline{G}_{X_{r}},\overline{G}_{Z_{r}}\}$,
    \item $\mathcal{T}_0' \subset \mathcal{T}_0G_{X_{\mathcal{T}}}$, where $G_{X_{\mathcal{T}}}= \langle\overline{G}_{X_{1}},\cdots, \overline{G}_{X_{y_I}}\rangle$.
\end{itemize}

\begin{corollary}\label{coro:dist_gauge_to_cbit_general}
If $\mathrm{min~wt}(\mathcal{G'}^{(n)}{G_{X_{\mathcal{T}}}^{(n)}}^*) \geq d(C_{\mathrm{EAOA}})$, then $d(C_{\mathrm{GGF}}) \geq d(C_{\mathrm{EAOA}})$ with equality being satisfied when $\mathcal{G}_0^{(\mathcal{H}_S)} = \emptyset$, $\mathcal{T}_0' = \mathcal{T}_0 \langle\mathcal{G}_0^{(\mathcal{T})}\rangle$, and $[\mathcal{T}_0^{(n)}]_{\mathcal{Z(H)}}^{\mathcal{P}_n}$ is a group. %
\end{corollary}
\begin{proof}
    From \Cref{prop:distance_gauge_to_stab}, the construction in \Cref{sec:gauge_to_stabilizer} does not reduce the distance. The construction in \Cref{sec:GaugeToCosetTransversal} does not reduce the distance when $\mathrm{min~wt}(\mathcal{G'}^{(n)}{G_{X_{\mathcal{T}}}^{(n)}}^*) \geq d(C_{\mathrm{EAOA}})$. As the generators of $G_{X_{\mathcal{T}}}$ can be viewed as $\overline{G}_{X_i}$ in the modified gauge symplectic pairs, using \Cref{prop:dist_gauge_to_cbit} the result follows.
\end{proof}

\begin{restatable}{corollary}{}\label{coro:dist_gauge_to_cbit_stabilizer_general}
Based on an EAOAQEC code $C_{\mathrm{EAOA}} = (\mathcal{H}, \mathcal{S}, \mathcal{G}_0, \mathcal{L}_0, \mathcal{T}_0)$, let $C_{\mathrm{EAh}}^{(\mathcal{T}_0)} = (\mathcal{H}, \mathcal{S}, \emptyset, \mathcal{G}_0 \cup \mathcal{L}_0, \mathcal{T}_0)$ be the EA hybrid subspace stabilizer code. If $d(C_{\mathrm{EAh}}^{(\mathcal{T}_0)}) = d(C_{\mathrm{EAOA}})$, then $d(C_{\mathrm{GGF}}) \geq d(C_{\mathrm{EAOA}})$, with equality being satisfied when $\mathcal{G}_0^{(\mathcal{H}_S)} = \emptyset$, $\mathcal{T}_0' = \mathcal{T}_0G_{X_{\mathcal{T}}}$, and $[\mathcal{T}_0^{(n)}]_{\mathcal{Z(H)}}^{\mathcal{P}_n}$ is a group.
\end{restatable}
The proof of \Cref{coro:dist_gauge_to_cbit_stabilizer_general} follows from \Cref{coro:dist_gauge_to_cbit}. 

\begin{remark}
    Using the proposed construction, although codes can be constructed without the need for additional entanglement-assistance, using additional ebits increases the number of gauge qubits and the code could have improved code properties such as better distance, lower weight stabilizer generators (retains the LDPC nature of the code), etc., which could be useful for practical applications.
\end{remark}

We next provide an example of the code constructed using the procedure described. 

\begin{example}[EAOAQEC code based on shortened Hamming code]\label{ex:EAOAQEC_ex}

We start with an EA subsystem code $C_{\mathrm{EAOA}}$, with the following stabilizer, gauge, and logical generators obtained from the classical shortened Hamming code ($n=15$ before shortening), and whose detailed construction is provided in Appendix \ref{app:example_construction}:

\[
\begin{array}{c|cccccccccc|cc}
\hhline{=============}
S_1^{\mathrm{init}} & Z & I & I & Z & Z & I & I & Z & Z & I & Z & I\\ 
S_2^{\mathrm{init}} & I & Z & I & Z & I & Z & I & Z & I & Z & I & Z\\ 
S_3^{\mathrm{init}} & I & I & Z & Z & Z & Z & I & I & I & I & I  & I\\ 
S_4^{\mathrm{init}} & I & I & I & I & I & I & Z & Z & Z & Z & I  & I\\ 
S_5^{\mathrm{init}} & X & I & I & X & X & I & I & X & X & I & X  & I\\ 
S_6^{\mathrm{init}} & I & X & I & X & I & X & I & X & I & X & I  & X\\ 
S_7^{\mathrm{init}} & I & I & I & I & I & I & X & X & X & X & I  & I\\ 
S_8^{\mathrm{init}} & I & I & X & X & X & X & I & I & I & I & I  & I\\ 
\hhline{-------------}
\overline{L}_X & I & X & X & I & X & I & X & X & I & I & I & I\\ 
\overline{L}_Z & I & Z & Z & I & Z & I & Z & Z & I & I & I & I\\ 
\hhline{-------------}
 
\overline{G}_{1} & I & X & I & I & I & I & X & I & I & X & I & I\\ 
\overline{G}_{2} & Z & Z & Z & Z & I & I & I & I & I & I & I & I\\
\overline{G}_{3} & X & I & X & I & X & I & I & I & I & I & I & I\\ 
\overline{G}_{4} & I & I & I & Z & Z & I & I & Z & Z & I & I & I\\ 
\overline{G}_{5} & I & X & I & X & X & I & I & I & I & I & I & I\\ 
\overline{G}_{6} & Z & I & Z & Z & I & I & I & Z & Z & I & I & I\\
\hhline{=============}
\end{array}
\]

We note that the gauge generators form symplectic pairs, namely, $(\overline{G}_{2i-1}, \overline{G}_{2i})$, for $i\in\{1,2,3\}$. Let $\mathcal{S}$ be the stabilizer group generated by $S_j^{\mathrm{init}}$, for $j\in\{1, \ldots , 8\}$. We note that the gauge symplectic pairs $(\overline{G}_3,\overline{G}_4)$ and $(\overline{G}_5,\overline{G}_6)$ can be modified to symplectic pairs $(\overline{G}_3,\overline{G}_4\overline{G}_6)$ and $(\overline{G}_3\overline{G}_5,\overline{G}_6)$. From the symplectic pair $(\overline{G}_1,\overline{G}_2)$, we add $\overline{G}_2$ to the stabilizer generating set and $\overline{G}_1$ becomes a coset transversal element. We add both the elements of the symplectic pair $(\overline{G}_3,\overline{G}_4\overline{G}_6)$ to the stabilizer generating set $\mathcal{S}$ to form $\mathcal{H'}$. Let $\mathcal{S'}$ be the stabilizer generating set of the new code based on $\mathcal{H'}$. We note that $\overline{G}_3$ and $\overline{G}_4\overline{G}_6$ would be extended with non-identity operators to form stabilizers in $\mathcal{S}'$; hence, $\overline{G}_3$ and $\overline{G}_4\overline{G}_6$ extended with identity operators become coset transversal elements. The operators in the gauge symplectic pair $(\overline{G}_3\overline{G}_5,\overline{G}_6)$ remain gauge generators in the new code.

We obtain coset representative set $\mathcal{T}_0'$ by adding 
$T_1^{(12)}=\overline{G}_{1}$, $T_2^{(12)}=S_7^{\mathrm{init}}\overline{L}_X\overline{G}_1\overline{G}_3$ (same coset as $\overline{G}_1$), and
$T_3^{(12)}=\overline{L}_Z\overline{G}_2\overline{G}_4$ (same coset as $\overline{G}_4\overline{G}_6$) to $\mathcal{T}_0 = \{T_0^{(12)}=I\}$. We recall that $T_j^{(12)}$ is the restriction of $T_j$ to the first 12 qubits with $T_j$ acting trivially on the rest of the qubits. We use the notation $T_j^{(12)}$ here as these operators will be extended below.

The EAOAQEC code $C_{\mathrm{GGF}}$ defined by $\mathcal{H}'$ whose symplectic subgroup $\mathcal{H}'_S = \langle h_1, h_2, h_3, h_4, h_5, h_6\rangle$ with symplectic pairs $(h_1, h_4)$, $(h_2, h_5)$, and $(h_3, h_6)$,  and isotropic subgroup $\mathcal{H}'_I = \langle h_7, h_8, h_9, h_{10}, h_{11}\rangle$, has group generators as defined in the table below:

\[
\begin{array}{c|cccccccccc}
\hhline{===========}
h_1 & Z & I & I & Z & Z & I & I & Z & Z & I \\ 
h_2 & I & Z & I & Z & I & Z & I & Z & I & Z \\
h_3 & Z & I & Z & I & Z & I & I & I & I & I \\ 
h_4 & X & I & I & X & X & I & I & X & X & I \\ 
h_5 & I & X & I & X & I & X & I & X & I & X \\
h_6 & X & I & X & I & X & I & I & I & I & I \\ 
h_7 & I & I & I & I & I & I & X & X & X & X \\ 
h_8 & I & I & X & X & X & X & I & I & I & I \\ 
h_9 & I & I & I & I & I & I & Z & Z & Z & Z \\ 
h_{10} & I & I & Z & Z & Z & Z & I & I & I & I \\ 
h_{11} & Z & Z & Z & Z & I & I & I & I & I & I \\ 
\hhline{===========}
\end{array},
\]

In the table below, we provide the $S_j^{\mathrm{init}}$ and $\overline{G}_j$ from which $h_i$ are obtained:
\[
\begin{array}{|c|c|c|c|c|c|c|c|c|c|c|c|}
\hline
&&&&&&&&&&&\\[-1em]
\text{Generator of }\mathcal{H}' & h_1 & h_2 & h_3 & h_4 & h_5 & h_6 & h_7 & h_8 & h_9 & h_{10} & h_{11} \\[1pt]
\hline
&&&&&&&&&&&\\[-1em]
\text{Generator of }\mathcal{G} & S_1^{\mathrm{init}} & S_2^{\mathrm{init}} & \overline{G}_4\overline{G}_6 & S_5^{\mathrm{init}} & S_6^{\mathrm{init}} & \overline{G}_3 & S_7^{\mathrm{init}} & S_8^{\mathrm{init}} & S_4^{\mathrm{init}} & S_3^{\mathrm{init}} & \overline{G}_2 \\[1pt]
\hline
\end{array}
\]

For the (hybrid) EAOAQEC code, we can take a subset of coset representatives $\mathcal{T}_0' = \{ T_0=I, T_1, T_2, T_3\}$, where the elements of $\mathcal{T}_0'$ are defined in the table below:
\[
\begin{array}{c|ccccccccccccc}
\hhline{==============}
T_0 & I & I & I & I & I & I & I & I & I & I & I & I & I\\ 
T_1 & I & X & I & I & I & I & X & I & I & X & I & I & I\\ 
T_2 & X & I & I & I & I & I & X & I & X & I & I & I & I\\ 
T_3 & Z & I & I & I & I & I & Z & I & Z & I & I & I & I\\ 
\hhline{==============}
\end{array},
\]

The generators of $\mathcal{H}'_S$ and $\mathcal{H}'_I$ are extended using $3$ ebits to generate the stabilizer group $\mathcal{S}'$. The stabilizer generators $\{S_i\}_{i=1}^{11}$ of $\mathcal{S}'$, logical operators $\mathcal{L}_0 = \{ \overline{L}_X, \overline{L}_Z \}$, and gauge generators $\mathcal{G}_0 = \{ \overline{G}_X = \overline{G}_3\overline{G}_5, \overline{G}_Z = \overline{G}_6 \}$ of $C_{\mathrm{GGF}}$ are defined in the table below: 

\[
\begin{array}{c|cccccccccc|ccc}
\hhline{==============}
S_1 & Z & I & I & Z & Z & I & I & Z & Z & I & Z & I & I \\ 
S_2 & I & Z & I & Z & I & Z & I & Z & I & Z & I & Z & I \\ 
S_3 & Z & I & Z & I & Z & I & I & I & I & I & I & I & Z \\ 
S_4 & X & I & I & X & X & I & I & X & X & I & X & I & I\\ 
S_5 & I & X & I & X & I & X & I & X & I & X & I & X & I \\ 
S_6 & X & I & X & I & X & I & I & I & I & I & I & I & X \\ 
S_7 & I & I & I & I & I & I & X & X & X & X & I & I & I \\ 
S_8 & I & I & X & X & X & X & I & I & I & I & I & I & I \\ 
S_9 & I & I & I & I & I & I & Z & Z & Z & Z & I & I & I \\ 
S_{10} & I & I & Z & Z & Z & Z & I & I & I & I & I & I & I \\ 
S_{11} & Z & Z & Z & Z & I & I & I & I & I & I & I & I & I \\ 
\hhline{--------------}
\overline{L}_X & I & X & X & I & X & I & X & X & I & I & I & I & I \\ 
\overline{L}_Z & I & Z & Z & I & Z & I & Z & Z & I & I & I & I & I \\ 
\hhline{--------------}
\overline{G}_{X} & X & X & X & X & I & I & I & I & I & I & I & I & I \\ 
\overline{G}_{Z} & Z & I & Z & Z & I & I & I & Z & Z & I & I & I & I \\
\hhline{==============}
\end{array}
\]

The EAOAQEC code $C_{\mathrm{GGF}}$ obtained from $\mathcal{H}'$ and $\mathcal{T}_0'$ is a $[\![n=10, k=1, d=3; r=1, e=3, |\mathcal{T}_0'|=4]\!]$ code with $m = 11$, $l=8$, and $s = 5$ that encodes one logical qubit and 2 classical bits; i.e., 4 classical bit strings. %
In the table below, we provide the syndromes of the elements of $\mathcal{T}_0'$ with the rows and columns corresponding to elements of $\mathcal{T}_0'$ and the generators of $\mathcal{S}'$, respectively:

\[
    \begin{array}{c||ccccccccccc}
    & S_1 & S_2 & S_3 & S_4 & S_5 & S_6 & S_7 & S_8 & S_9 & S_{10} & S_{11} \\
    \hhline{============}
    T_0 & 0 & 0 & 0 & 0 & 0 & 0 & 0 & 0 & 0 & 0 & 0\\
    \hhline{------------}
    T_1 & 0 & 0 & 0 & 0 & 0 & 0 & 0 & 0 & 0 & 0 & 1\\
    \hhline{------------}
    T_2 & 0 & 0 & 1 & 0 & 0 & 0 & 0 & 0 & 0 & 0 & 1\\
    \hhline{------------}
    T_3 & 0 & 0 & 0 & 0 & 0 & 1 & 0 & 0 & 0 & 0 & 0\\
    \hhline{============}
    \end{array}
\]

We note that changing $\mathcal{T}_0'$ could change the number of classical bit strings encoded and the distance of the code, while the other parameters remain the same. For instance, we construct two EAOAQEC codes with the same stabilizer generators, logical operators, and gauge generators as above but add either $T_4$ or $T_5$ to $\mathcal{T}_0'$, where $T_4$ and $T_5$ are from the below table:
\[
\begin{array}{c|ccccccccccccc}
\hhline{==============}
T_4 & XZ & X & I & I & I & I & Z & I & XZ & X & I & I & I\\ 
T_5 & X & I & I & I & I & I & X & X & I & I & I & I & I\\ 
\hhline{==============}
\end{array}
\]
Observe that $T_4 = T_1T_2T_3$, while $T_5$ is a coset transversal element of the EA subsystem code based on $\mathcal{S}$.
The syndromes of $T_4$ and $T_5$ are given in the below table:
\[
    \begin{array}{c||ccccccccccc}
    & S_1 & S_2 & S_3 & S_4 & S_5 & S_6 & S_7 & S_8 & S_9 & S_{10} & S_{11} \\
\hhline{============}
    T_4 & 0 & 0 & 1 & 0 & 0 & 1 & 0 & 0 & 0 & 0 & 0\\
\hhline{------------}
    T_5 & 0 & 1 & 1 & 0 & 0 & 0 & 0 & 0 & 0 & 0 & 1\\
\hhline{============}
    \end{array}
\]
Further note that adding $T_4$ to $\mathcal T_0'$ yields a code with identical parameters to $\mathcal C_{\mathrm{GGF}}$ except that the new code encodes an additional classical bit string.
On the other hand, adding $T_5$ to $\mathcal T_0'$ instead of $T_4$ gives a code with an additional encoded bit string but with  a reduced distance of 2.

\end{example}

We next enumerate a few special cases of the construction described above: 
\begin{itemize}
    \item[(a)] From subsystem codes: We choose the code $C_{\mathrm{EAOA}}$ to be an $[[n,k,d;r]]$ subsystem code to obtain $[[n,k,d';r-y+e,e,c_b]]$-EAOAQEC codes. An example of an $[[n,k,d';r-y,0,2^y]]$-OAQEC code obtained from $C_{\mathrm{EAOA}}$ involves choosing $\mathcal{T}_0' = \langle\{\overline{G}_{X_i}\}_{i=1}^y\rangle$ and adding $\{\overline{G}_{Z_i}\}_{i=1}^y$ to $\mathcal{S}$ to obtain $\mathcal{H}' = \langle\mathcal{S},\{\overline{G}_{Z_i}\}_{i=1}^y\rangle$. We note that, for this case, the OAQEC code has the same distance as $C_{\mathrm{EAOA}}$ when $\mathrm{min~wt}({\mathcal{G}'}^{(n)}{G_{X_{\mathcal{T}}}^{(n)}}^*) \geq d(C_{EAOA})$. An $[[n,k,d'';r-y+e,e,2^y]]$-EAOAQEC code can be obtained by adding some $\overline{G}_{X_i}\overline{G}_{Z_j}$s to $\mathcal{H}'$ instead of $\overline{G}_{Z_j}$, where $i,j \in\{1,\dots,y\}$ and $e = \mathrm{log}_2(|\mathcal{H}'|/|Z(\mathcal{H}')|)$. Further note that $\mathcal{T}_0'$ can be chosen to be a subset of $\langle\{\overline{G}_{X_i}\}_{i=1}^y\rangle$, for example, $\{\overline{G}_{X_i}\}_{i=1}^y$. 
    \item[b)] From EA subsystem codes: We choose the code $C_{\mathrm{EAOA}}$ to be an $[[n,k,d;r,e]]$ EA subsystem code from which an EAOAQEC code can be constructed. \Cref{ex:EAOAQEC_ex} constructs an EAOAQEC code from an EA subsystem code.
    \item[c)] From hybrid subsystem codes: We choose the code $C_{\mathrm{EAOA}}$ to be an $[[n,k,d;r,0,c_b]]$ hybrid subsystem code from which an EAOAQEC code can be constructed. Examples of such constructions are similar to those provided for constructions from subsystem codes with considering $\mathcal{T}_0'$ to be a subset of $\mathcal{T}_0G_{X_{\mathcal{T}}}$.
    \item[d)] From EA hybrid subsystem codes: We choose the code $C_{\mathrm{EAOA}}$ to be an $[[n,k,d;r,e,c_b]]$ EA hybrid subsystem code; i.e., an EAOAQEC code, from which another EAOAQEC code is constructed. This construction could be used to improve the properties of the code such as distance and code rate.
\end{itemize}

We finish by applying the construction to the subsystem color code. 

\begin{example}[15-qubit subsystem color code]
For the $[\![15,1,3;6]\!]$-subsystem color code, let us consider $y=2$, $\mathcal{G}_0^{[\mathcal{H}_S]} = \{\overline{G}_{X_1}\overline{G}_{Z_2},\overline{G}_{X_2}\}$, $\mathcal{G}_0^{[\mathcal{H}_I]} = \emptyset$, and $G_{X_{\mathcal{T}}} = \langle\overline{G}_{Z_1},\overline{G}_{Z_2}\rangle$. Let the coset transversal subset be $\mathcal{T}_0' = \{I^{\otimes 15}, \overline{G}_{Z_1}, \overline{G}_{Z_2}, \overline{G}_{Z_1}\overline{G}_{Z_2}\}$. The EAOAQEC code $C_{\mathrm{GGF}}$ obtained is an $[\![15,1,3;4,1,4]\!]$ code that encodes $2$ classical bits.

We next consider the $[[15,1,2;6,0,3]]$-hybrid subsystem code  based on the 15-qubit subsystem color code with $\mathcal{T}_0 = \{I^{\otimes 15}, X_5Z_6, X_9Z_{11}\}$. 
\[
\begin{array}{c|ccccccccccccccc}
\hhline{================}
T_0 & I & I & I & I & I & I & I & I & I & I & I & I & I & I & I\\
T_1 & I & I & I & I & X & Z & I & I & I & I & I & I & I & I & I\\
T_2 & I & I & I & I & I & I & I & I & X & I & Z & I & I & I & I\\
\hhline{================}
\end{array}
\]
The EAOAQEC code obtained considering $\mathcal{G}_0^{[\mathcal{H}_S]} = \{\overline{G}_{X_1}\overline{G}_{Z_2},\overline{G}_{X_2}\}$, $\mathcal{G}_0^{[\mathcal{H}_I]} = \emptyset$, $G_{X_{\mathcal{T}}} = \langle\overline{G}_{Z_1},\overline{G}_{Z_2}\rangle$, and $\mathcal{T}_0' = \{I^{\otimes 15}, \overline{G}_{Z_1}, \overline{G}_{Z_2}, \overline{G}_{Z_1}\overline{G}_{Z_2}\} \{I^{\otimes 15}, X_5Z_6, X_9Z_{11}\}$ is a $[\![15, 2, 2; 4,1,12]\!]$-EAOAQEC code that encodes $12$ classical bit strings.
\end{example}

\section{Conclusion} \label{sec:conc}

The EAOAQEC framework for entanglement-assisted quantum codes unifies the original frameworks of EAQEC \cite{brun2006correcting}, EAOQEC \cite{hsieh2007general}, and EACQ \cite{kremsky2008classical} by viewing them through an algebraic approach enabled by the recently introduced stabilizer formalism for OAQEC \cite{kribs2023stabilizer}. We established conditions that say when a set of errors is correctable for a given code, generalizing the previous EA error correction theorems. We then used the theorem to define distances for these codes, and we proved bounds for various subclasses of them. The EAOAQEC approach also evidently leads to new classes of hybrid EA subspace codes, we showed exactly how the EACQ codes sit as a subclass of such codes, and EA subsystem codes, where we gave examples and constructions in both the hybrid and non-hybrid cases. 

This work generates several questions and potentially new lines of inquiry. Indeed, in principle one could consider corresponding extensions of any scenario in which EA codes have found application, of which there appear to be several as a simple literature search confirms. To name a few, first note that one could consider more general errors beyond those which are noiseless for Bob, coming from different physically motivated scenarios, with logical encoded $e$bits to deal with errors of specific types (an example is given in \cite{lai2012entanglement} with errors considered there as `Bob storage errors'). Here we have focused on systems of qubits, but we fully expect that the whole framework readily extends to qudits (the work \cite{nadkarni2021non} and the qudit extension in \cite{kribs2023stabilizer} can be used as guidance in that respect). One could also consider an extension of the framework to infinite-dimensional Hilbert spaces and the von Neumann algebra commuting operator framework, for both mathematical motivations \cite{beny2009quantum,crann2016private,crann2020state,conlon2023quantum} and the potential applicability to settings of relevance in error correction for photonic quantum computing \cite{gottesman2001encoding,Bourassa2021blueprintscalable} and the entanglement structure of black holes~\cite{hayden_preskill2007,hayden_penintgon2020,hayden_penington2019,yoshida2021recovery}.
Further, we wonder what the approaches of catalytic quantum codes \cite{brun2014catalytic} and entanglement-assisted LDPC codes \cite{hsieh2011high}  applied to the new codes introduced here could yield, as well as possible connections with recent hybrid subsystem code applications and constructions such as \cite{TanNemHybrid}.   
And of course the generalization of the design of encoding circuits, syndrome computation circuits, and decoding techniques \cite{EASC_Encoding,OSC_Encoding,NBEASC_Encoding,NBEASC_SyndromeComputation,EASC_DecodingComplexity} tailored to this new general entanglement-assisted framework is important for the practical implementation of the codes designed using the framework. Finally, one could explore the notion of punctured quantum codes~\cite{Rains1999, gundersen2024puncturingquantumstabilizercodes} in the EAOAQEC framework.

We plan to undertake some of these investigations elsewhere, and we invite other interested researchers to do so as well. 

\strut 

{\noindent}{\it Acknowledgements.}
S.A. was partly supported by a Mitacs Accelerate grant. D.W.K. was partly supported by NSERC Discovery Grant 400160. Research at Perimeter Institute is supported in part by the Government of Canada through the Department of Innovation, Science and Economic Development Canada and by the Province of Ontario through the Ministry of Colleges and Universities. We thank Tarik El-Khateeb for helping us with \Cref{fig:EAOAQEC Diagram}.

\bibliographystyle{plainurl}

\bibliography{refs}

\clearpage
\appendix

\section{15-qubit subsystem color code}\label{app:3D_Subsystem_Color_code}
We consider a $[\![15,1,3;6, 0, 1]\!]$ subsystem color code \cite{Paetznick_Reichardt2013,Bombin2015Gauge} with distance $3$
whose stabilizer and gauge generators are provided in the table below:

\[
\begin{array}{c|ccccccccccccccc}
\hhline{================}
S_1 & X & I & X & I & X & I & X & I & X & I & X & I & X & I & X\\
S_2 & X & X & I & I & X & X & I & I & X & X & I & I & X & X & I\\
S_3 & I & I & I & X & X & X & X & I & I & I & I & X & X & X & X\\
S_4 & I & I & I & I & I & I & I & X & X & X & X & X & X & X & X\\
S_5 & Z & I & Z & I & Z & I & Z & I & Z & I & Z & I & Z & I & Z\\
S_6 & Z & Z & I & I & Z & Z & I & I & Z & Z & I & I & Z & Z & I\\
S_7 & I & I & I & Z & Z & Z & Z & I & I & I & I & Z & Z & Z & Z\\
S_8 & I & I & I & I & I & I & I & Z & Z & Z & Z & Z & Z & Z & Z\\
\hhline{----------------}
\overline{G}_{X_1} & I & I & X & I & I & I & X & I & I & I & X & I & I & I & X\\
\overline{G}_{Z_1} & I & I & I & I & I & I & I & I & I & I & I & Z & Z & Z & Z\\
\overline{G}_{X_2} & I & I & I & I & I & I & I & I & I & I & I & X & X & X & X\\
\overline{G}_{Z_2} & I & I & Z & I & I & I & Z & I & I & I & Z & I & I & I & Z\\
\overline{G}_{X_3} & I & I & I & I & X & I & X & I & I & I & I & I & X & I & X\\
\overline{G}_{Z_3} & I & I & I & I & I & I & I & I & I & Z & Z & I & I & Z & Z\\
\overline{G}_{X_4} & I & I & I & I & I & I & I & I & I & X & X & I & I & X & X\\
\overline{G}_{Z_4} & I & I & I & I & Z & I & Z & I & I & I & I & I & Z & I & Z\\
\overline{G}_{X_5} & I & I & I & I & I & X & X & I & I & I & I & I & I & X & X\\
\overline{G}_{Z_5} & I & I & I & I & I & I & I & I & Z & I & Z & I & Z & I & Z\\
\overline{G}_{X_6} & I & I & I & I & I & I & I & I & X & I & X & I & X & I & X\\
\overline{G}_{Z_6} & I & I & I & I & I & Z & Z & I & I & I & I & I & I & Z & Z\\
\hhline{================}
\end{array}
\]

\section{Code Construction Algorithms from \texorpdfstring{\Cref{sec:isotropictosymplectic}}{Section 7.2}}\label{isotosymalgs}

\begin{algorithm}[ht]
\caption{Obtain EAOAQEC code from a hybrid subsystem code by converting isotropic operators to extended symplectic operators}\label{alg:EAOAQEC_HybridSubsystemCode}
\begin{algorithmic}
\Require Generators of $\mathcal{S}$ and $\mathcal{T}_0$ 
\Ensure $\mathcal{H} = \langle H_1,\dots,H_m\rangle$ and $\mathcal{T}_0'^{(n-f)}$\vspace{-0.4cm}
\AlgoProcedure{\vspace{-0.4cm}
\begin{itemize}
    \item[1)] Let $E_Q$ denote the set of qubit indices such that no Pauli operator whose support is a subset of $E_Q$ belongs to $\mathcal{Z(S)}$. Let $f = |E_Q|$. For obtaining an EAOAQEC code with f ebits with $1\leq f < \mathrm{min~wt}(\mathcal{Z}(\mathcal{S}))$, $E_Q$ can be chosen to be any $f$ qubit indices. 
    \item[2)] Obtain the check matrix $H = [H_1|H_2]$ of the stabilizer code by stacking the binary representation of stabilizer generators. Construct $H'$ from $H$ by permuting the columns of $H_1$ and $H_2$ corresponding to the qubit indices in $E_Q$ to the first $2f$ columns. Perform partial Gaussian elimination on $H'$ by considering the first $2f$ columns to be the pivotal columns. Permute back the columns of the modified $H'$ to that of the order of qubits in $H$ to obtain $H^{(G)} = [H_1^{(G)}|H_2^{(G)}]$.
    \end{itemize}}
    \algstore{code_const_algo}
\end{algorithmic}
\end{algorithm}

\begin{algorithm}                     
\begin{algorithmic}[1]                   %
\algrestore{code_const_algo}
\AlgoProcedureNoTitle{
    \begin{itemize}
    \item[3)] Consider the $2f$ stabilizer generators $S_j$ corresponding to the pivotal rows in the Gaussian elimination procedure to be the extended symplectic pairs and the rest stabilizer generators to correspond to the extended isotropic operators. The qubits with indices in $E_Q$ correspond to the ebits and the rest $(n-f)$ qubits correspond to Alice's qubits. $\mathcal{H}$ is obtained from the restriction of $S_j$ to the qubit indices not in $E_Q$.
    \item[4)] Multiply each element $T$ of $\mathcal{T}_0$ with the corresponding stabilizer generators $S_j$ that have non-identity operators on the ebits to obtain the product $T'$ to have identity on the ebits. $\mathcal{T}_0^{(n-f)}$' is obtained to be the restriction of $T'$ to the qubit indices not in $E_Q$.
\end{itemize}
\Return $\mathcal{H}$ and $\mathcal{T}_0'^{(n-f)}$}
\end{algorithmic}
\end{algorithm}

\begin{algorithm}[H]
\caption{Obtain EA operator algebra CSS code from a hybrid subsystem CSS code based on a dual-containing classical code}\label{alg:EAOAQEC_HybridSubsystemCSSCode}
\begin{algorithmic}
\Require Generators of $\mathcal{S}$ and $\mathcal{T}_0$ 
\Ensure $\mathcal{H}$ and $\mathcal{T}_0'^{(n-f)}$ 
\AlgoProcedure{
\begin{itemize}
    \item[1)] Obtain the check matrix $H = \left[\begin{array}{c|c}
        H & 0\\ 0 & H
    \end{array}\right]$ of the CSS code.
    \item[2)] Using Gaussian elimination, convert the $H$ into row reduced echelon form. Let $H^{(G)}$ be the matrix obtained after Gaussian elimination procedure. Let $S_j$s be the stabilizer generators corresponding to the check matrix $H = \left[\begin{array}{c|c}
        H^{(G)} & 0\\ 0 & H^{(G)}
    \end{array}\right]$. Let $f = \mathrm{rank}(H)$. %
    \item[3)] The qubits corresponding to the pivotal columns in the Gaussian elimination procedure are the ebits while the rest are Alice's qubits. %
    $\mathcal{H}$ is obtained from the restriction of elements of $S_j$ to Alice's qubit indices.
    \item[4)] Multiply each element $T$ of $\mathcal{T}_0$ with the corresponding stabilizer generators $S_j$ that have non-identity operators on the ebits to obtain the product $T'$ to have identity on the ebits. $\mathcal{T}_0'^{(n-f)}$ is obtained to be the restriction of $T'$ to Alice's qubit indices.
\end{itemize}
\Return $\mathcal{H}$ and $\mathcal{T}_0'^{(n-f)}$}
\end{algorithmic}
\end{algorithm}

\section{Construction of the EAOAQEC code in \texorpdfstring{\Cref{sec:GaugeToCosetTransversal_general}}{Section 7.4}}\label{app:example_construction}

The EAOAQEC code in \Cref{sec:GaugeToCosetTransversal_general} is constructed from the $[15,11,3]$ Hamming code $C_H$ defined by the following parity check matrix:
\begin{align}
    H_{H} = \left[\begin{array}{cccccccccccccccc}
        0 & 0 & 0 & 0 & 0 & 0 & 0 & 1 & 1 & 1 & 1 & 1 & 1 & 1 & 1\\ 
        0 & 0 & 0 & 1 & 1 & 1 & 1 & 0 & 0 & 0 & 0 & 1 & 1 & 1 & 1\\
        0 & 1 & 1 & 0 & 0 & 1 & 1 & 0 & 0 & 1 & 1 & 0 & 0 & 1 & 1\\
        1 & 0 & 1 & 0 & 1 & 0 & 1 & 0 & 1 & 0 & 1 & 0 & 1 & 0 & 1 
    \end{array}\right]
\end{align}
The Hamming code $C_H$ is first shortened by removing bits at locations 1, 12, 13, 14, and 15 to obtain a $[10,6,3]$ code $C_{HS}$ with the following parity check matrix:
\begin{align}
    H_{HS} = \left[\begin{array}{ccccccccccc}
        0 & 0 & 0 & 0 & 0 & 0 & 1 & 1 & 1 & 1\\ 
        0 & 0 & 1 & 1 & 1 & 1 & 0 & 0 & 0 & 0\\
        1 & 1 & 0 & 0 & 1 & 1 & 0 & 0 & 1 & 1\\
        0 & 1 & 0 & 1 & 0 & 1 & 0 & 1 & 0 & 1 
    \end{array}\right]
\end{align}
The last row of $H_{HS}$ is added to the third row of $H_{HS}$ to obtain the following parity check matrix $H_{HS}^{(f)}$ of $C_{HS}$:
\begin{align}
    H_{HS}^{(f)} = \left[\begin{array}{ccccccccccc}
        0 & 0 & 0 & 0 & 0 & 0 & 1 & 1 & 1 & 1\\ 
        0 & 0 & 1 & 1 & 1 & 1 & 0 & 0 & 0 & 0\\
        1 & 0 & 0 & 1 & 1 & 0 & 0 & 1 & 1 & 0\\
        0 & 1 & 0 & 1 & 0 & 1 & 0 & 1 & 0 & 1 
    \end{array}\right]
\end{align}

While the last two rows of $H_{HS}$ are not self-orthogonal, the first two rows are self-orthogonal and each pair of rows are orthogonal to each other. Thus, using two entangled bits, an EA stabilizer code can be constructed with the following isotropic subgroup and symplectic subgroup:
\begin{align*}
    \mathcal{H}_I^{(1)} = \langle \bar{Z}_4,\bar{Z}_5,\bar{Z}_6,\bar{Z}_7 \rangle,\\
        \mathcal{H}_S^{(1)} = \langle \bar{Z}_1,\bar{Z}_2,\bar{X}_1,\bar{X}_2 \rangle,
\end{align*}
where the generators of $\mathcal{H}_I^{(1)}$ and $\mathcal{H}_S^{(1)}$ are given in the below table:
\[
\begin{array}{c|ccccccccccccccc}
\hhline{================}
\bar{Z}_1 & Z & I & I & Z & Z & I & I & Z & Z & I\\
\bar{Z}_2 & I & Z & I & Z & I & Z & I & Z & I & Z \\
\bar{X}_1 & X & I & I & X & X & I & I & X & X & I\\
\bar{X}_2 & I & X & I & X & I & X & I & X & I & X\\
\bar{Z}_4 & I & I & I & I & I & I & X & X & X & X\\
\bar{Z}_5 & I & I & X & X & X & X & I & I & I & I\\
\bar{Z}_6 & I & I & I & I & I & I & Z & Z & Z & Z\\
\bar{Z}_7 & I & I & Z & Z & Z & Z & I & I & I & I\\
\hhline{================}
\end{array}
\]

Let $C(\mathcal{H}_I^{(1)}, \mathcal{H}_S^{(1)})$ be the EA quantum code obtained from $\mathcal{H}_I^{(1)}$ and $\mathcal{H}_S^{(1)}$.
We next choose four symplectic pairs of gauge and logical operators of $C(\mathcal{H}_I^{(1)}, \mathcal{H}_S^{(1)})$ of the form $(\bar{X}_i^{(\mathrm{init})}, \bar{Z}_i^{(\mathrm{init})})$, namely
\begin{align}
    (\bar{X}_1^{(\mathrm{init})}, \bar{Z}_1^{(\mathrm{init})}) &= (X_2X_7X_{10}, Z_1Z_2Z_3Z_4),\nonumber\\
    (\bar{X}_2^{(\mathrm{init})}, \bar{Z}_2^{(\mathrm{init})}) &= (X_1X_3X_5,Z_4Z_5Z_8Z_9),\nonumber\\
    (\bar{X}_3^{(\mathrm{init})}, \bar{Z}_3^{(\mathrm{init})}) &= (X_2X_4X_5,Z_1Z_3Z_4Z_8Z_9),\nonumber\\
    (\bar{X}_4^{(\mathrm{init})}, \bar{Z}_4^{(\mathrm{init})}) &= (X_2X_3X_5X_7X_8,Z_2Z_3Z_5Z_7Z_8).\nonumber
\end{align}
We note that the pairs $(\bar{X}_i^{(\mathrm{init})}, \bar{Z}_i^{(\mathrm{init})})$ of gauge and logical operators of $C(\mathcal{H}_I^{(1)},\mathcal{H}_S^{(1)})$ are based on the dual of the classical shortened Hamming code $C_{HS}$.

    For the EA subsystem code, we initially choose $(\bar{X}_1^{(\mathrm{init})}, \bar{Z}_1^{(\mathrm{init})})$, $(\bar{X}_2^{(\mathrm{init})}, \bar{Z}_2^{(\mathrm{init})})$, and $(\bar{X}_3^{(\mathrm{init})}, \bar{Z}_3^{(\mathrm{init})})$ as gauge operators and $(\bar{X}_4^{(\mathrm{init})}, \bar{Z}_4^{(\mathrm{init})})$ as the logical operator of the code.
    To construct the EAOAQEC code, we add $\bar{Z}_2^{(\mathrm{init})}\bar{Z}_3^{(\mathrm{init})} = Z_1Z_3Z_5$, $\bar{X}_2^{(\mathrm{init})}$, and $\bar{Z}_1^{(\mathrm{init})}$ to $\mathcal{H}$. To $\mathcal{T}_0$, we add $\bar{X}_1^{(\mathrm{init})}$, $\bar{X}_1^{(\mathrm{init})}\bar{X}_2^{(\mathrm{init})}\bar{X}_4^{(\mathrm{init})}\overline{Z}_4 = X_1X_7X_9$ (same coset as $\bar{X}_1^{(\mathrm{init})}\bar{X}_2^{(\mathrm{init})}$), and $\bar{Z}_1^{(\mathrm{init})}\bar{Z}_2^{(\mathrm{init})}\bar{Z}_4^{(\mathrm{init})} = Z_1Z_7Z_9$ (same coset as $\bar{Z}_1^{(\mathrm{init})}\bar{Z}_2^{(\mathrm{init})}$). Thus, we obtain the isotropic and symplectic subgroup generators and elements of the coset transversal subset $\mathcal{T}_0$ to be the following:
\[
\begin{array}{c|ccccccccccccccc}
\hhline{================}
\bar{Z}_3 & Z & I & Z & I & Z & I & I & I & I & I\\
\bar{X}_3 & X & I & X & I & X & I & I & I & I & I \\
\bar{Z}_8 & Z & Z & Z & Z & I & I & I & I & I & I\\
\hhline{================}
T_0 & I & I & I & I & I & I & I & I & I & I & I & I & I\\ 
T_1 & I & X & I & I & I & I & X & I & I & X & I & I & I\\ 
T_2 & X & I & I & I & I & I & X & I & X & I & I & I & I\\ 
T_3 & Z & I & I & I & I & I & Z & I & Z & I & I & I & I\\ 
\hhline{==============}
\end{array}
\]

We note that $\bar{Z}_3$ and $\bar{X}_3$ anticommute with each other while commute with all other elements in the isotropic and symplectic subgroups. The operator $\bar{Z}_8$ commutes with all elements in the isotropic and symplectic subgroups. The example constructed in \Cref{sec:GaugeToCosetTransversal_general} is based on these updated isotropic and symplectic subgroups and $\mathcal{T}_0$.
\end{document}